\documentclass[onecolumn,12pt,journal,drafts]{IEEEtran}
\usepackage[comma,numbers,square,sort&compress]{natbib}
\usepackage{amsmath}
\usepackage{amssymb}
\usepackage{float}
\usepackage{amsthm}
\usepackage{graphicx}
\usepackage{subfigure}
\usepackage{epstopdf}
\usepackage{color}
\usepackage{verbatim}
\usepackage{tikz,times}
\usepackage{algorithm}
\usepackage{algorithmic}
\usepackage{diagbox}
\usepackage{hyperref}
 \DeclareMathOperator{\blockdiag}{blockdiag}
 \DeclareMathOperator{\diag}{diag}

  \DeclareMathOperator{\Expect}{\mathbb{E}}
   
 \newcommand{\norm}[1]{\left\lVert#1\right\rVert}
 
 \makeatletter
 \DeclareRobustCommand{\rvdots}{%
 	\vbox{
 		\baselineskip4\p@\lineskiplimit\z@
 		\kern-\p@
 		\hbox{.}\hbox{.}\hbox{.}
 	}}
 	\makeatother

\theoremstyle{plain}
\newtheorem{theorem}{Theorem}[section]
\newtheorem{proposition}{\textbf{Proposition}}[section]
\newtheorem{lemma}{\textbf{Lemma}}[section]

\newtheorem{remark}{\textbf{Remark}}[section]

\theoremstyle{definition}
\newtheorem{definition}{\textbf{Definition}}[section]
\newtheorem{assumption}{\textbf{Assumption}}[section]

\allowdisplaybreaks[4]

\begin{document}

\title{Event-Triggered  Distributed  Estimation 
	With Decaying Communication Rate}

\author{Xingkang He, Yu Xing, Junfeng Wu, Karl H. Johansson
	\thanks{X. He, Y. Xing and K. H. Johansson are with the Division of Decision and Control Systems, School of Electrical Engineering and Computer Science. KTH Royal Institute of Technology, SE-100 44 Stockholm, Sweden (xingkang@kth.se, yuxing2@kth.se, kallej@kth.se).}
	\thanks{J. Wu is with the School of Data Science, Chinese University of Hong Kong (Shenzhen), Shenzhen, Guangdong Province, P. R. China; College of Control Science and Engineering, Zhejiang University, Hangzhou, Zhejiang Province, P. R. China (jfwu@zju.edu.cn).}}

\maketitle

\begin{abstract}
\boldmath
We study  distributed estimation of a high-dimensional static parameter vector through a group of sensors whose communication network is modeled by a fixed directed  graph. Different from   existing time-triggered communication schemes,  an event-triggered asynchronous   scheme is investigated in order to reduce communication   while preserving estimation convergence. 
A  distributed estimation algorithm with a single step size is first proposed based on 
an event-triggered communication  scheme  with a time-dependent decaying threshold. 
With the event-triggered scheme, each sensor  sends   its estimate to neighbor sensors only when the  difference between the current estimate and the last sent-out estimate  is larger than the triggering threshold. Different  sensors can have different  step sizes and triggering thresholds, enabling the parameter estimation process to be conducted in a fully distributed way.  
We prove that  the proposed algorithm has mean-square   and almost-sure convergence respectively, {under an integrated condition of sensor network topology and sensor measurement matrices. The condition is satisfied if the topology is a balanced digraph containing a spanning tree and the system is collectively observable. }  
The collective observability is the possibly mildest condition, since it is a   spatially and temporally  
collective condition of all sensors and allows  sensor measurement matrices to be time-varying, stochastic, and non-stationary.  
Moreover, we provide   estimates for the   convergence rates, which are related to the  step size as well as the triggering threshold.  Furthermore, as an essential metric of  sensor communication intensity in the event-triggered distributed algorithms, 
the communication rate is proved to decay to zero with a certain speed almost surely as time goes to infinity.  In addition, we show that it is feasible to tune the threshold and the step size such that   requirements of algorithm convergence and  communication rate decay are satisfied simultaneously. {  We also show that given the step size, adjusting the decay speed of the triggering threshold can lead to a tradeoff between the convergence rate of the estimation error and the decay speed of the communication rate.
	Specifically, increasing the decay speed of the threshold 
	would make the communication rate decay faster, but reduce the  convergence rate of the estimation error.	
} 
Numerical simulations are provided to illustrate the developed results.
\end{abstract}


\section{Introduction}\label{sec:intro}
Networked  systems  monitoring complex environments generate a tremendous amount of data, which can be used in the estimation of system states or unknown parameters. 
Parameter estimation is  one of the most important tools  in control theory, signal processing, and machine learning  with  extensive applications in sensor networks,  weather prediction,  cyber-physical systems, environmental monitoring, transportation, etc.  
The development of computational and energy efficient algorithms, which are able to handle imperfect data from sensor networks, is drawing more and more attention.

{\subsection{Motivation and Challenges}
	To estimate a high-dimensional parameter vector (such as temperature over a large environment) is usually infeasible by a local algorithm for each sensor only with local measurements.
	Thus,   it is necessary to design a collaborative estimation algorithm for sensors to obtain accurate estimates of the parameter.
	Centralized and  distributed architectures are the two dominating approaches to sensor network estimation.  
	In the  centralized architecture,   a data center   runs  estimation   algorithms based on all sensor data in order to estimate a state or system parameter.
	In this way,   optimality  in certain sense can be ensured, such as the Kalman filter achieving the minimum variance unbiased estimate for linear dynamical systems under Gaussian noise.   
	However, with  large-scale deployment of sensors, the centralized architecture is no longer efficient.  Therefore, we need to develop effective and efficient distributed estimation   algorithms  based on local sensor data and local sensor communication.   }

{In   distributed parameter estimation,   messages shared between sensors play   important roles. 
	Since  sensor measurements only contain partial parameter information, it is not sufficient to design a global parameter estimation algorithm only by communicating sensor measurements. 
	Quite a few  distributed estimators are proposed by requiring that sensors send their estimates to neighbor sensors at \emph{every} time instant of measurement sampling. }
Diffusion-based distributed estimators are proposed in \cite{Cattivelli2008Diffusion} and 
\cite{cattivelli2010diffusionLMS} by  solving least-squares and least-mean-squares problems over networks.
In \cite{rad2010distributed}, a distributed estimator based on sensor beliefs  is proposed over a strongly connected communication network. In \cite{kar2011convergence},   a gossip distributed estimator is proposed to handle random  failures of network links. The results are extended in \cite{kar2012distributed} to nonlinear systems under imperfect communication channels, and  extended to more general system models and communication networks in \cite{kar2013distributed}, where an adaptive learning method is proposed in order to achieve   asymptotic efficiency.   
A robust distributed parameter estimator is given in \cite{zhang2012distributed} for  systems with Markovian switching networks and uncertain measurement models. 
{However, since  the above references require that each sensor communicates with its neighbors persistently,  there will be severe energy consumption and  serious channel congestion if  the measurement sampling rate is   high or the parameter is   high-dimensional. 
	Under constrained communication resources, although the above algorithms can  still work by  requiring that sensors communicate periodically, their estimation performance could be much degraded.

	Therefore, it is vital to develop a performance-guaranteed distributed estimation algorithm based on an effective and efficient sensor communication mechanism. Nevertheless, there are two  challenges. First, for  potentially time-varying, stochastic, and non-stationary  sensor measurement models  with weak observability (e.g., collective observability), how to guarantee the  performance of the distributed algorithm under limited communications is challenging. Second, under the influence of noisy measurements,   how to quantify  the communication frequency of sensors is difficult.  
	
}



\subsection{Related Work}
To mitigate the issue of sensor communications, there are some   approaches, including data quantization, compressed sensing, adaptive  sampling, and event-triggered communications. 
In the following discussion, we focus our attention on   adaptive  sampling and event-triggered communications, since they can make better use of posterior information compared to the former two approaches.

In parameter estimation under  constraint on the total sensing effort, adaptive sampling   helps  improve    performance by strategically allocating sensing effort in   future data collection based on the information extracted from data collected previously. 
In \cite{haupt2011distilled}, a sequential adaptive sampling-and-refinement procedure called distilled sensing is proposed  to detect and estimate  parameters. 
The reference \cite{castro2015adaptive} investigates how to estimate the support set of a  sparse parameter  by making full use of structural information. 
For a class of physical-constrained sensing models, the limitations and advantages of adaptive sampling are analyzed in \cite{davenport2016constrained}, and it is shown that a constrained adaptive sampling method can substantially improve estimation performance. Moreover,    
the adaptive sampling problem for a class of continuous-time Markov state processes is studied in \cite{rabi2012adaptive}. The above and related literature require either  structural information or prior distribution of parameters, which  is difficult to  satisfy in monitoring complex environments.  
In addition, many methods in the literature  are introduced in the centralized architecture. It is  difficult to extend these methods to the distributed architecture with large-scale sensor networks and high-dimensional system parameters. 

In contrast, an event-triggered scheme is suitable to the distributed architecture under constrained communication resource, since it can efficiently determine when   a sensor should share data  to other sensors    without the prior  knowledge of parameter structure. In such a way,     event-triggered sensor communications are usually asynchronous and aperiodic, thus different from   traditional time-triggered communications.
A number of centralized estimators with  event-triggered measurement schedulers are proposed in the literature, such as \cite{sinopoli2004kalman,wu2012event} for   state estimation of  dynamical systems,  and  \cite{you2013asymptotically,han2015optimal,diao2018event} for   static parameter estimation. 
Regarding the event-triggered parameter estimation,  the authors in \cite{you2013asymptotically} propose a measurement scheduler such that the asymptotic estimation performance is optimized.  A stochastic measurement scheduler is studied in \cite{han2015optimal} to compensate for the loss of the Gaussianity of the system, which ensures the maximum-likelihood
parameter estimator.  An event-triggered scheduler for finite impulse response systems with binary measurements is proposed in \cite{diao2018event} where communication rate is analyzed. 
The above centralized event-triggered measurement schedulers are usually not suitable to the distributed architecture, because  scheduling local measurements between neighboring sensors is  insufficient to design a global parameter estimator when each sensor can only observe partial elements of the parameter vector.
In the distributed architecture for event-triggered  estimation,
there are several methods  for dynamical systems. 
{For example, an LMI approach is studied in \cite{muehlebach2017distributed} for the cooperative estimation and control of a dynamical system under an event-triggered communication protocol.
	An event-triggered distributed Kalman filter is proposed in \cite{battistelli2018distributed} and proved to be stable in terms of means-square boundedness of the estimation error in each sensor.
	For linear dynamical systems under state equality constraints, an  event-triggered projected distributed Kalman filter is studied in \cite{he2019distributed} with guaranteed estimation error stability under collective system observability. For more related works of event-triggered distributed estimation for dynamical systems, we refer the readers to \cite{ge2019distributed} and the references therein.
}
Another related topic is event-triggered distributed optimization, which concerns the distributed design of optimization algorithms such that the global optimal solution is reached by each computational node. In this direction, most algorithms are proposed in noise-free or disturbance-free settings and analyzed to { show  the influence of  event-triggered mechanisms on  algorithm convergence \cite{cao2020decentralized,li2017event}}. An event-triggered property for continuous-time systems, which is called Zeno behavior (i.e., an infinite number of events occur in a finite amount of
time), is also of interest in some papers \cite{chen2016event}.  
However, there are few event-triggered distributed optimization or estimation algorithms handling imperfect data from  stochastic environments.
In addition,  the above studies analyze the influence of  event-triggered threshold to estimation performance,  but the tradeoff  between  communication frequency and  estimation performance is not well established   in the distributed architecture. To the best  of our knowledge, there is no result of distributed parameter estimation on communication rate, which is an  essential  metric of sensor network communication level.





\subsection{Contributions}
In this paper, we study the  event-triggered  distributed parameter estimation problem  for the sake of reducing sensor communication   while preserving convergence of estimators. 
The main contributions  are summarized in the following:

1. 
We propose a recursive event-triggered distributed parameter algorithm with a single step size  (Algorithm \ref{alg:A}) for a group of sensors with noisy measurements. 
The algorithm has several advantages: First, it is fully distributed in the sense that 
each sensor only relies on the local information and has its own step size and triggering threshold.
Second, without requiring exact noise distribution or statistics, 
the event-triggered scheme enables  sensors to communicate in an efficient way that each sensor  sends   its estimate to neighboring sensors only when the  difference between the current estimate and the last sent-out estimate  is larger than the triggering threshold.  Third, the algorithm is  scalable to large-scale sensor networks, since  its update is independent of the network size.  




2.  
We prove that  the proposed algorithm with a properly designed step size and triggering threshold  achieves mean-square  convergence (Theorem \ref{thm_mean_square_convergence}).
In addition, we provide the estimate for the   convergence rate and establish its connection to  network structure and  system observability.  
{The results are obtained under an integrated condition of  sensor network topology and sensor measurement matrices. The condition is satisfied if the topology is a balanced digraph containing a spanning tree and the system is collectively observable.}
The collective observability is the possibly mildest condition, since it  allows each sensor to observe partial elements of the parameter, 
sensor measurement matrices to be time-varying, stochastic, and non-stationary,   
as well as the noise processes to be martingale difference sequences.  Under some extra conditions on step size and triggering threshold, we prove the algorithm's output   is asymptotically convergent to the true parameter almost surely (a.s.) with an estimated convergence rate (Theorem~\ref{thm_conver_rate}).

%
%

3.  We  prove that the communication rate is decaying to zero almost surely  as time   goes to infinity  (Theorem \ref{thm_rate}) with a quantified speed.  
Moreover, we show that it is feasible to tune the threshold and the step size such that the requirements of the algorithm convergence and the communication rate decay are satisfied simultaneously. {We also show that given the step size, adjusting the decay speed of the triggering threshold can lead to a tradeoff between the convergence rate of the estimation error and the decay speed of the communication rate.
	Specifically, increasing the decay speed of the threshold 
	would make the communication rate decay faster, but reduce the  convergence rate of the estimation error.} To the best  of our knowledge, this is the first result on communication rate in the direction of event-triggered distributed algorithms for estimating  static parameters or solutions \cite{cao2020decentralized,chen2016event}.
Our result indicates that  while ensuring successful parameter estimation, the algorithm enables to   alleviate  channel burden and resource consumption in sensor communication compared to   existing time-triggered approaches \cite{rad2010distributed,kar2011convergence,kar2013distributed,cattivelli2010diffusionLMS,zhang2012distributed,wang2019distributed,lei2015distributed} which require a persistent positive communication rate.

The results of this paper are significantly different from the literature.   Regarding the step size, we remove the requirement that all sensors share the same setting \cite{kar2011convergence,kar2013distributed,kar2012distributed,wang2019distributed,zhang2012distributed,chen2020resilient,lei2015distributed,he2019distributed_CCC}, as well as the requirements of the  
concrete forms \cite{kar2011convergence,kar2013distributed,kar2012distributed,chen2020resilient,lei2015distributed,he2019distributed_CCC} and the monotonicity in \cite{wang2019distributed}.  
The stationarity condition of measurement matrices in \cite{zhang2012distributed} is removed.    
Since sensor measurements and measurement matrices are not shared between neighbors, our algorithm can be more suitable to the scenarios with privacy requirement than the diffusion estimators in  \cite{Cattivelli2008Diffusion,Cattivelli2010Distributed}.
Moreover, we generalize  sensor communication topologies from undirected graphs \cite{kar2011convergence,kar2013distributed,kar2012distributed,he2019distributed_CCC}  to a class of  directed graphs, such as balanced digraphs containing a
spanning tree.      


\subsection{Paper Organization}
{In Section 2, we formulate the considered problem of event-triggered distributed parameter estimation and introduce some graph preliminaries, the communication model, and the  mathematical model of sensors. In order to solve this problem, in Section 3 an event-triggered distributed estimation algorithm is   proposed and analyzed in terms of estimation error convergence  (mean-square and almost-sure convergence) and communication rate. For reading convenience, the proofs for the results in Section 3 are provided in Section 4. Numerical simulations in two examples are provided in  Section 5 to illustrate the developed results.  Section 6 concludes the paper. }


\textbf{Notations.}  Denote 
$\mathbb{R}^{n\times m}$  the set of real-valued matrices with $n$ rows and $m$ columns, with $\mathbb{R}^n=\mathbb{R}^{n\times 1}$ and $\mathbb{R}^1=\mathbb{R}$. Let $\mathbb{R}^+$ and $\mathbb{N}^+$ be the sets of positive real-valued scalars and integers, respectively, with $\mathbb{N}=\mathbb{N}^+\cup\{0\}$. Denote $\Omega$ and $\emptyset$  the universal set and empty set, respectively.
$I_{n}$ stands for the $n$-dimensional square identity matrix. 
$\textbf{1}_n$ stands for the $n$-dimensional vector with all elements being one. The superscript ``$\sf T$" represents the transpose. 
$\Expect\{x\}$ denotes the mathematical expectation of the random variable $x$, and  $\blockdiag\{\cdot\}$   represents the diagonalizations of block elements.  
$A\otimes B$ is the Kronecker product of $A$ and $B$.  $\norm{x}$ is the 2-norm of a vector $x$, and $\norm{A}$ is the induced norm, i.e., $\norm{A}=\sup_{\norm{x}=1}\norm{Ax}$, where $A\in \mathbb{R}^{n\times m}$, $x\in\mathbb{R}^m$. The mentioned scalars, vectors and matrices of this paper are all real-valued. Let $\sigma(\cdot)$ be the $\sigma$-algebra operator which generates the smallest $\sigma$-algebra. For a real-valued matrix or vector sequence $\{a(t)\}$ and a real number sequence $\{b(t)\}$, the operator $a(t)=O(b(t))$ means that  there is a constant $c\geq 0$, such that for each element sequence of $\{a_i(t)\}$, $\lim\limits_{t\rightarrow\infty}|a_i(t)/b(t)|\leq c$, and the operator $a(t)=o(b(t))$ means  $\lim\limits_{t\rightarrow\infty}|a_i(t)/b(t)|=0$.

\section{Preliminaries and Problem Formulation}\label{sec:problem}

{In this section, we formulate the problem of event-triggered distributed parameter estimation and introduce some graph preliminaries, the communication model, and the  mathematical model of sensors. }

\subsection{Graph Preliminaries}
In this paper, the communication between $N$ sensors of a network is modeled as a digraph $\mathcal{G=(V,E,A)}$, where
$\mathcal{V}=\{1,2,\dots,N\}$ is the node set,  $\mathcal{E}\subseteq \mathcal{V}\times \mathcal{V}$ is the edge set. A directed edge  $(i,j)\in \mathcal{E}$, if and only if there is a communication link from $j$ to $i$, where  $j$ is called the parent node and $i$ is called the child node. The matrix $\mathcal{A}=[a_{i,j}]_{i,j=1}^N$ is the 
weighted  adjacency matrix with $a_{i,j}\geq 0$, where  $a_{i,j}>0$, if $(i,j)\in \mathcal{E}$.     The parent neighbor set and child neighbor set of node $i$ are denoted by $\{j\in\mathcal{V}|(i,j)\in \mathcal{E}\}\triangleq\mathcal{N}_{i}$ and $\{j\in\mathcal{V}|(j,i)\in \mathcal{E}\}\triangleq\mathcal{N}_{i}^c$, respectively. 
Suppose that the graph has no self loop, which means $a_{i,i}=0$ for any $i\in\mathcal{V}$.
$\mathcal{G}$ is called a balanced digraph, if $\sum_{j=1}^{N}a_{i,j}=\sum_{i=1}^{N}a_{j,i}$ for all $i\in\mathcal{V}$. $\mathcal{G}$ is called an undirected graph, if $\mathcal{A}$ is symmetric. The Laplacian matrix of $\mathcal{G}$ is denoted by $\mathcal{L}=\mathcal{D}-\mathcal{A}$, where  $\mathcal{D}=\diag\{\sum_{j=1}^{N}a_{1,j},\dots,\sum_{j=1}^{N}a_{N,j}\}$.  
The mirror graph of the digraph  $\mathcal{G}$ is an undirected
graph, denoted by   $\bar{\mathcal{G}}=\mathcal{(V,E_{\bar{\mathcal{G}}},A_{\bar{\mathcal{G}}})}$ with $\mathcal{A}_{\bar{\mathcal{G}}}=[\bar a_{i,j}]_{i,j=1}^N$,
$\bar a_{i,j}=\bar a_{j,i}=(a_{i,j}+a_{j,i})/2$ (\cite{olfati2004consensus}).
$\mathcal{G}$ is called strongly connected if   for any pair nodes $(i_{1},i_{l})$, $l> 1,$ there exists a   path from $i_{l}$ to $i_{1}$ consisting of edges in the set $\{(i_{m},i_{m+1})\in \mathcal{E}|m=1,2,\dots,l-1\}$. We call  $\mathcal{G}$ is connected if it is strongly connected and undirected.
A directed tree is a digraph, where each node
except the root has exactly one parent node. A spanning tree of $\mathcal{G}$
is a directed tree whose node set is $\mathcal{V}$ and whose edge set is a
subset of $\mathcal{E}$.

\subsection{Problem Setup}
Consider an unknown high-dimensional parameter vector $\theta\in \mathbb{R}^M$  observed by $N>0$ sensors with the following model
\begin{equation}\label{system_all}
y_{i}(t)= H_{i}(t)\theta+v_{i}(t),i=1,2,\dots,N,
\end{equation}
where 
$y_{i}(t)\in \mathbb{R}^{m_{i}}$ is the measurement vector, $v_{i}(t)\in \mathbb{R}^{m_{i}}$   is the measurement noise, and $H_{i}(t)\in \mathbb{R}^{m_{i}\times M}$  represents the known   measurement matrix of sensor $i$, all at time $t$.

The communication rate is an essential metric of  communication intensity in the event-triggered distributed algorithms. Its mathematical definition over the digraph $\mathcal{G}$ is in the following.
\begin{definition}\label{def_com_rate}
	For the   digraph $\mathcal{G}$, in a given time interval  $[0,t]\cap \mathbb{N}$, 	 the \textbf{communication rate} $\lambda_c(t)$   is given by
	\begin{align}\label{eq_rate}
	\lambda_c(t)=\frac{\sum_{i\in\mathcal{V}}K_i(t)|\mathcal{N}_{i}^c| }{t \sum_{i\in\mathcal{V}}|\mathcal{N}_{i}^c|},
	\end{align}
	where $K_i(t)$ is the accumulated triggering (data-sending) times of node $i$ in $[0,t]\cap \mathbb{N}$, and $|\mathcal{N}_{i}^c|$ is the child neighbor number of node $i$.
\end{definition}

According to this definition, $	\lambda_c(t)\in[0,1]$. When $\lambda_c(t)\equiv 1$,   communication occurs all the time, which is the case for the time-triggered distributed estimation algorithms \cite{rad2010distributed,kar2011convergence,kar2013distributed,cattivelli2010diffusionLMS,zhang2012distributed,wang2019distributed}.  
When $\lambda_c(t)\equiv 0$, there is no sensor communication over the network.

The problem considered in this paper is how to design a distributed estimation algorithm  with 
an event-triggered communication scheme, and find conditions  such that  the output of the algorithm   is asymptotically convergent to the   parameter vector at a certain rate while the communication rate of the sensor network   is decaying to zero as time   goes to infinity.

\section{Main results}\label{sec:results}
{The formulated problem  is solved in  this section. First, an event-triggered distributed   algorithm is proposed to estimate the parameter.  Then,   the mean-square and almost-sure convergence  of the algorithm are studied, respectively. Additionally,   the decay speed of the communication rate   is analyzed.  The proofs of these results are provided in Section~\ref{sec:proofs}.}

\subsection{Event-triggered distributed estimation}
{To solve the problem posed in Section~\ref{sec:problem},} we propose the event-triggered distributed estimation algorithm in Algorithm \ref{alg:A} for each sensor $i\in\mathcal{V}$, where the initial estimate $x_{i}(0)=x_{i,0}\in\mathbb{R}^M$ is fixed. 	An  example of Algorithm \ref{alg:A} for a sensor network with 7 nodes   is provided in FIG. \ref{fig:diag}.

\begin{algorithm}[t]
	\caption{Event-Triggered Distributed Estimation}
	\label{alg:A}
	\begin{algorithmic}[1]
		\STATE {\textbf{Input:} $ x_{i}(0),\{y_i(t)\}_{t\geq 0},\{H_i(t)\}_{t\geq 0},\{\alpha_i(t)\}_{t\geq 0},\{f_i(t)\}_{t\geq 0}$}\\		\vskip 2pt
		\STATE {\textbf{Output:} $\{x_{i}(t)\}_{t\geq 0}$}\\		\vskip 2pt
		\FOR{$t=0,1,\dots$} 
		\STATE {//\textbf{ Data transmission}}\\		\vskip 2pt
		
		\IF{$t=0$} 
		\STATE {Send $x_i(0)$ to each child neighbor sensor $j\in\mathcal{N}_i^c$ and let $x_i(\tau_{0})=x_i(0)$ and $k_i(0)=1$.}
		\ELSIF{$ \|x_i(t) - x_i(\tau_{k_i(t-1)})\| > f_i(t)$}
		\STATE {Send $x_i(t)$ to each child neighbor sensor $j\in\mathcal{N}_i^c$ and let $k_i(t)=k_i(t-1)+1$ and $\tau_{k_i(t)}:=t$.}
		\ELSE
		\STATE{$k_i(t)=k_i(t-1)$. }
		\ENDIF
		\vskip 10pt
		\STATE {//\textbf{ Data receiving}}\\		\vskip 2pt
		\FOR{$j\in\mathcal{N}_i$} 
		\IF{Receive estimate $x_j(t)$ from parent neighbor sensor $j\in\mathcal{N}_i$}
		\STATE{Let   $k_j(t)=k_j(t-1)+1$ and $\tau_{k_j(t)}:=t.$}
		\ELSE
		\STATE{$k_j(t)=k_j(t-1).$}
		\ENDIF
		\ENDFOR
		\vskip 10pt
		\STATE {// \textbf{Estimate update}}	 
		\begin{align}\label{eq_estimator0}
		\begin{split}
		x_{i}(t+1)=&x_{i}(t)+\alpha_i(t)H_{i}^{\sf T}(t)( y_{i}(t)-H_{i}(t)x_{i}(t))\\
		&+\alpha_i(t)\sum_{j\in\mathcal{N}_{i}}a_{i,j}(x_{j}(\tau_{k_j(t)})-x_{i}(t)).
		\end{split}
		\end{align}
		\ENDFOR
	\end{algorithmic}
\end{algorithm}

\begin{remark}
	We have a few remarks on Algorithm \ref{alg:A}.
	\begin{enumerate}
		\item The notation $x_i(t)$ denotes the estimate of $\theta$ by sensor $i$ at time $t.$ The notation $\tau_{k_i(t)}$ denotes the time of $k_i(t)\text{-th}$ triggering instant  till time $t$. The amount of $k_i(t)$  shows how many events have been triggered  for   sensor $i$ till time $t$. 		
		The scalar $\alpha_i(t)$ is the step size of sensor $i$ at time $t$, and the scalar $f_i(t)\geq 0,$ $i\in\mathcal{V}$, is the  event-triggered threshold of sensor $i$ at time $t$. 
		Both $\{\alpha_i(t)\}_{t\geq 0}$ and $\{f_i(t)\}_{t\geq 0}$ are to be designed.
		\item The event-triggered scheme $ \|x_i(t) - x_i(\tau_{k_i(t-1)})\| > f_i(t)$ is to  determine whether the estimate   $x_i(t)$ is worth sharing with child nodes by comparing it with the last sent-out estimate $ x_i(\tau_{k_i(t-1)})$.  
		{		Since the estimate $x_i(t)$ provides the global parameter information, in weak observability conditions the scheme outweighs the existing event-triggered  schemes transmitting local measurements \cite{you2013asymptotically,han2015optimal,diao2018event}. Moreover, compared to the existing schemes \cite{sinopoli2004kalman,wu2012event,you2013asymptotically,han2015optimal},  the proposed scheme is built on a more general stochastic framework without requiring knowledge of accurate noise distribution or statistics.}
		
		\item {To ensure the convergence of Algorithm~\ref{alg:A}, the triggering threshold $f_i(t)$ should decay to zero fast enough, as required in Assumption~\ref{asmp_consistency2}. However, to avoid that sensors communicate frequently all the time,  $f_i(t)$ should not decay too fast, as required in Assumption~\ref{ass_triggering}. Remark~\ref{rem_step} will show  that it is feasible to satisfy the assumptions on  $f_i(t)$ simultaneously. Moreover, as shown in Remark~\ref{rem_tradeoff},   the decay speed of $f_i(t)$ can lead to a tradeoff between the decaying speed of the communication rate and the convergence rate of the estimation error.}

		\item  Regarding the time instants \{$\tau_{k_i(t)}$\}, it holds that  $	\tau_{k_i(t)}=\inf_{t>\tau_{k_i(t-1)}} t$, if $ \|x_i(t) - x_i(\tau_{k_i(t-1)})\| > f_i(t)$, otherwise $\tau_{k_i(t)}=\tau_{k_i(t-1)}$.
		Thus, $\{\tau_{k_i(t)}\}$ is  determined only by past and current information.

	\end{enumerate}
\end{remark}

%

\begin{remark}
	We have a few remarks on the advantages of Algorithm \ref{alg:A} by comparing to  existing algorithms.
	\begin{enumerate}
		\item   An advantage of Algorithm \ref{alg:A} comparing to   the diffusion  estimation algorithms in \cite{Cattivelli2008Diffusion,Cattivelli2010Distributed}  is that Algorithm \ref{alg:A} does not require the local measurements and measurement matrices to be shared between sensors. Thus, Algorithm \ref{alg:A} is more suitable to the scenarios with privacy requirement and limited communication bandwidth. 
		
		
		\item 
		Algorithm \ref{alg:A}  does not require any global   knowledge of the system, and thus is fully distributed.  For example,  it removes the requirement in \cite{kar2011convergence,kar2013distributed,kar2012distributed,wang2019distributed,zhang2012distributed,he2019distributed_CCC} that all sensors share the same step size. Moreover, 		
		Algorithm \ref{alg:A} is able to handle open sensor networks  where some sensors may break down or new sensors are plug-in, 	which however is   intractable for the algorithms \cite{kar2011convergence} requiring the total sensor number and  the measurement matrices of all sensors. 
		
		%
		%
		%
		\item 		
		Since Algorithm \ref{alg:A}  can tremendously reduce the redundant transmissions of  estimates,  it can require less communications than  the existing time-triggered distributed algorithms   \cite{wang2019distributed,zhang2012distributed,kar2011gossip,cattivelli2010diffusionLMS,kar2013distributed} for convergence.

	\end{enumerate}
\end{remark}

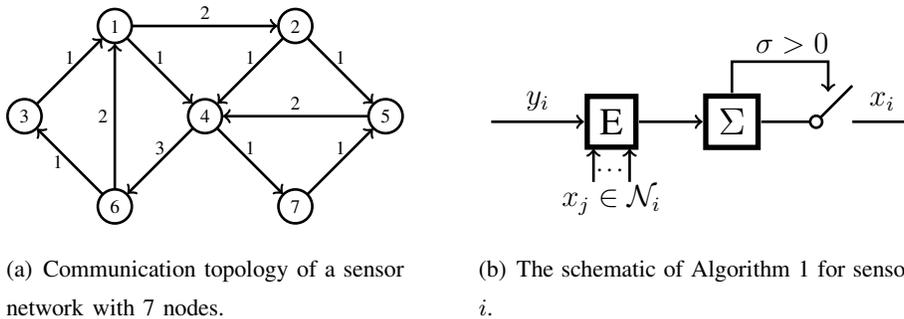
\begin{figure*}[t]
	\centering
	\subfigure[\label{fig:graph_ill} Communication topology of a sensor network with 7 nodes.]{\begin{tikzpicture}[scale=0.6, transform shape,line width=1pt]
		\node [draw,shape=circle,line width=1pt,minimum size=0.7cm] (1) at (0, 0) {1};
		\node [draw,shape=circle,line width=1pt,minimum size=0.7cm] (2) at (4, 0) {2};
		\node [draw,shape=circle,line width=1pt,minimum size=0.7cm] (3) at (-2, -2) {3};
		\node [draw,shape=circle,line width=1pt,minimum size=0.7cm] (4) at (2, -2) {4};
		\node [draw,shape=circle,line width=1pt,minimum size=0.7cm] (5) at (6, -2) {5};
		\node [draw,shape=circle,line width=1pt,minimum size=0.7cm] (6) at (0, -4) {6};
		\node [draw,shape=circle,line width=1pt,minimum size=0.7cm] (7) at (4, -4) {7};
		\path [->] (1) edge[line width=1pt] node[above] {2}  (2);
		\path [->] (1) edge[line width=1pt] node[above] {1}  (4);
		\path [->] (2) edge[line width=1pt] node[above] {1}  (4);
		\path [->] (2) edge[line width=1pt] node[above] {1}  (5);
		\path [->] (3) edge[line width=1pt] node[above] {1}  (1);
		\path [->] (4) edge[line width=1pt] node[above] {3}  (6);
		\path [->] (4) edge[line width=1pt] node[above] {1}  (7);
		\path [->] (5) edge[line width=1pt] node[above] {2}  (4);
		\path [->] (6) edge[line width=1pt] node[left] {2}  (1);
		\path [->] (6) edge[line width=1pt] node[left] {1}  (3);
		\path [->] (7) edge[line width=1pt] node[above] {1}  (5);
		\end{tikzpicture}}
	\qquad
	\subfigure[\label{fig:event_ill}The schematic of Algorithm \ref{alg:A} for sensor $i$. ]{	\begin{tikzpicture}[scale=0.8, transform shape,line width=1pt]
		\draw (0,0)  node[rectangle,draw,scale=1.5,line width=2pt] (est_i)   {E};
		\draw (2,0)  node[rectangle,draw,scale=1.5,line width=2pt] (est_j)   {$\Sigma$};
		\path [->] (est_i) edge   (est_j);	
		\draw (3.4,0) node[circle,draw,scale=0.5,line width=1pt] (emp1)   {};
		\path [-] (est_j) edge (emp1);
		\path [-] (emp1) edge (4,0.6);
		\path [->] (4,0) edge  node[left,above,scale=1.2] {$x_i$} (5,0);
		\path [->] (-2,0) edge node[left,above,scale=1.2] {$y_i$} (est_i);
		\path [->] (-0.3,-1) edge (-0.3,-0.45);
		\path [->] (0.3,-1)  edge (0.3,-0.45);
		\path   (-0.2,-0.8) -- node[auto=false]{\ldots} (0.2,-0.8);
		\draw (0,-1.3)  node[auto=false,scale=1.2] (est)   {$x_j\in\mathcal{N}_i$};
		\path [draw] (est_j) -- (2,1) -- (3.6,1) -- (3.6,0.5) [->];	
		\draw (3,1.3)  node[auto=false,scale=1.2] (sigma)   {$\sigma>0$};
		\end{tikzpicture}}
	\caption{ An example of Algorithm \ref{alg:A} for a sensor network with 7 nodes.
		The communication  of sensors  forms a weighted balanced digraph with a
		spanning tree, as shown in (a) where the numbers on edges stand for weights. The schematic of Algorithm \ref{alg:A} for sensor $i\in\{1,2,\dots,7\}$  is shown in (b), consisting of  a parameter estimator $E$	 
		and a communication scheduler $\Sigma$.
		The estimator has inputs of local measurement $y_i$ and parent neighbor estimates $x_j\in\mathcal{N}_i$. The scheduler $\Sigma$  sends estimate $x_i$   to child neighbor nodes when the event is triggered, i.e., $\sigma>0$, where $\sigma:=\|x_i(t) - x_i(\tau_{k_i(t-1)})\|-f_i(t)$.	 
	}
	\label{fig:diag}
\end{figure*}

To ease the notation, we write $\tau_{k_i(t)}$ as $\tau_k^i$ in the below text, i.e.,  $\tau_{k}^i:=\tau_{k_i(t)}$. The subscript $k$ of $\tau_{k}^i$ is kept to emphasize the number of triggering times, but the reader should keep in mind that $k$ in $\tau_k^i$ depends on time $t$ and sensor $i$. With this definition, $x_i(\tau_{k}^i):=x_i(\tau_{k_i(t)})$.
Then we rewrite (\ref{eq_estimator0}) in the following way
\begin{align}\label{eq_estimator02}
x_{i}(t+1)=&x_{i}(t)+\alpha_i(t)H_{i}^{\sf T}(t)( y_{i}(t)-H_{i}(t)x_{i}(t)) \\\nonumber
& +\alpha_i(t)\sum_{j\in\mathcal{N}_{i}}a_{i,j}(x_{j}(t)-x_{i}(t)) +\alpha_i(t)\sum_{j\in\mathcal{N}_{i}}a_{i,j}(x_{j}(\tau_{k}^j)-x_{j}(t))\nonumber.
\end{align}

%

\subsection{Convergence and convergence rate}
{In order to ensure that Algorithm \ref{alg:A} in the previous subsection provides accurate parameter estimates, in this subsection we find conditions such that  the output of  Algorithm \ref{alg:A}, i.e., $x_i(t)$, is asymptotically convergent to $\theta$ with an estimated convergence rate.}
To proceed, we introduce the following notations:
\begin{align}\label{eq_denotations}
\begin{split}
\bar{\alpha}(t) &= \blockdiag\left\{ \alpha_1(t) I_M, \dots, \alpha_N(t) I_M \right\},\\
\bar D_H(t)& = \blockdiag\left\{H_1^{\sf T}(t),\dots,H_N^{\sf T}(t)\right\},\\
\tau_k& = \left[\tau_k^1,\dots,\tau_k^N\right]^{\sf T},\qquad  f_{\max}(t)=\max_{i\in\mathcal{V}}f_i(t),\qquad 
f_{\min}(t)=\min_{i\in\mathcal{V}}f_i(t),\\
V(t) &= [v_1^{\sf T}(t),\dots,v_N^{\sf T}(t)]^{\sf T},\qquad\quad\quad
X(t)=\left[x_1^{\sf T}(t),\dots,x_N^{\sf T}(t)\right]^{\sf T},\\
X(\tau_k)&=\left[x_{1}^{\sf T}(\tau_{k}^1),\dots,x_{N}^{\sf T}(\tau_{k}^N)\right]^{\sf T},\qquad 
Y(t)=\left[y_1^{\sf T}(t),\dots,y_N^{\sf T}(t)\right]^{\sf T}.
\end{split}
\end{align}
Given the notations in (\ref{eq_denotations}), 
the compact form of (\ref{eq_estimator02}) is given in the following
\begin{equation}
\begin{aligned}\label{eq_estimation_compact}
X(t+1)=&X(t)-\bar{\alpha}(t)(\mathcal{L}\otimes I_M)X(t) + \bar{\alpha}(t)\bar D_H(t)(Y(t)-\bar D_H^{\sf T}(t)X(t))\\
&+\bar{\alpha}(t)(\mathcal{A}\otimes I_M)(X(\tau_{k})-X(t)).
\end{aligned}
\end{equation}

In the basic probability space $(\Omega,\mathcal{F},P)$, define the filtration  $\mathcal{F}(t) := \sigma(\bar{D}_H(s), V(s), 0 \le s \le t)$, $t\ge 0$, and $\mathcal{F}(-1):= \{\emptyset, \Omega\}$. The following assumption is needed in this paper.


\begin{assumption}\label{asmp_consistency2}~\\
	(i.a) There exists a sequence $\{\alpha(t)\}$ such that $\alpha_i(t)/\alpha(t) \to 1$ as $t \to \infty$ for all $1\le i \le N$. In addition, $\alpha(t) > 0$, $\alpha(t) \to 0$,  $\sum_{t=1}^{\infty} \alpha(t) = \infty$, and
	\begin{equation*}
	\frac{1}{\alpha(t+1)} - \frac{1}{\alpha(t)} \to \alpha_0 \ge 0.
	\end{equation*}
	(i.b) $\sum_{t=1}^{\infty} \alpha(t)^{2(1-\delta)} < \infty$ for some $\delta \in [0, 1/2)$.\\
	(ii.a) $\{\bar{D}_H(t)\}$ and $\{V(t)\}$ are independent sequences.\\
	(ii.b) $\{V(t)\}$ is a martingale difference sequence and there is a scalar $\rho>2$, such that  \[\sup_{t \in \mathbb{N}}\mathbb{E}\{\|V(t)\|^\rho|\mathcal{F}(t-1)\} := c_V< \infty, \text{ a.s}.\]
	(ii.c) $\sup_{t \in \mathbb{N}} \|\bar{D}_H(t)\|^2 \le D < \infty$ a.s. for some positive constant $D$, and there exist $h\in\mathbb{N}^+$, $\lambda \in \mathbb{R}^+$, such that for any $m\in\mathbb{N}$ and some $\delta \in [0,1/2)$,
	\begin{align}\label{eq_obser_condi}
	\lambda_{\min}\left[ \sum_{t=mh}^{(m+1)h-1}\mathbb{E}\left\{ \bar{\mathcal{L}} \otimes I_M + \bar D_H(t) \bar D_H^{\sf T}(t) |\mathcal{F}(mh-1)\right\}\right]\geq \lambda + h \alpha_0 \delta, ~\text{a.s.,}
	\end{align}
	where $\bar{\mathcal{L}}=(\mathcal{L} + \mathcal{L}^{\sf T})/2$.\\
	(iii.a) $f_{\max}(t)/\alpha^{\delta}(t) \to 0$, $t\to \infty$, for some $\delta \in [0, 1/2)$.\\
	(iii.b) $\sum_{t=1}^{\infty} \left(\big(\alpha(t)\big)^{1-\delta} f_{\max}(t) \right) < \infty$, for some $\delta \in [0, 1/2)$.
\end{assumption}

\begin{remark}
	Some remarks on Assumption \ref{asmp_consistency2} are given:
	\begin{itemize}
		\item In (i.a), $\alpha(t)$ provides a reference rate of step sizes. It is necessary that step sizes of sensors are in the same order. If there exist $i$ and $j$ such that $\alpha_i(t)/\alpha_j(t) \to 0$, then the information provided by sensor $i$ vanishes in terms of $\alpha_j(t)$. The requirements in (i.a) and (i.b) for choosing the step size $\alpha(t)$ are quite general, where $\alpha(t)=a/(t+1)$ for some $a>0$  in \cite{kar2012distributed} is a special case. 
		\item In (ii.c), \eqref{eq_obser_condi} indicates  network   connectivity and  system observability. 	Note that the eigenvalue on the left-hand side of \eqref{eq_obser_condi} is non-negative. It is positive under  proper conditions of network connectivity  and system observability, such as in   Proposition~\ref{prop_obser_condi}. In order to ensure a certain convergence rate of Algorithm \ref{alg:A}, 
		the eigenvalue has to be large enough. If the step size is slowly decreasing, such that $\alpha_0 = 0$, then the second term on the right-hand side of \eqref{eq_obser_condi} is zero. In other cases, since $\alpha_0$ depends only on the step size, one can design the latter so that $\alpha_0$ is small enough and \eqref{eq_obser_condi} holds.

		\item Condition (iii.a) requires decaying thresholds $\{f_i(t)\}$ to ensure that the estimates can still be shared over an infinite horizon. Otherwise, in a convergent algorithm (its output  may not converge to $\theta$), since the difference between two adjacent estimates is decreasing to zero, the events would no longer be triggered after some finite time.	
		\item  Condition (iii.b) is required for almost-sure convergence of the algorithm, meaning that the triggering thresholds have to decrease fast enough, to ensure certain convergence rate of the algorithm. 
	\end{itemize}
\end{remark}

{The following proposition shows that  the integrated condition \eqref{eq_obser_condi} in Assumption~\ref{asmp_consistency2} is satisfied under certain 
	network connectivity and system  observability.}
\begin{proposition}\label{prop_obser_condi}
	Suppose that $\mathcal{G}$  is a balanced digraph containing a
	spanning tree,  and there exist $h \in \mathbb{N}^+$ and $\tilde{\lambda} \in \mathbb{R}^+$ such that for any $m \in \mathbb{N}$,
	\begin{align}\label{con_obser}
	{\lambda_{\min}\left[\sum_{j=1}^N \sum_{t=mh}^{(m+1)h-1}\mathbb{E}\left\{ H_j^{\sf T}(t) H_j(t) \Bigg|\mathcal{F}(mh-1)\right\}\right] \ge \tilde{\lambda}, ~\text{a.s.,}}
	\end{align}
	then there exists $\beta \in \mathbb{R}^+$ such that for any $m \in \mathbb{N}$,
	\begin{align*}
	\lambda_{\min}\left[\sum_{t=mh}^{(m+1)h-1}\mathbb{E}\left\{ \bar{\mathcal{L}}\otimes I_M + \bar D_H(t) \bar D_H^{\sf T}(t) |\mathcal{F}(mh-1)\right\} \right] \ge \beta, ~\text{a.s.}
	\end{align*}
\end{proposition}

\begin{proof}
	Since $\mathcal{G}$ is a balanced digraph, 
	according to \cite[Th. 7]{olfati2004consensus}, $\bar{\mathcal{L}}=(\mathcal{L}^T+\mathcal{L})/2$ is the Laplacian matrix of the  undirected mirror graph $\bar{\mathcal{G}},$ i.e., $\mathcal{L}_{\bar{\mathcal{G}}}=\bar{\mathcal{L}}$.
	Then  by \cite[Th. 2.8]{Mesbahi2010Graph}, $\bar{\mathcal{L}}$ has a unique eigenvalue zero, and $\lambda_{2}(\bar{\mathcal{L}})>0$. 
	Hence for a unit vector $x \in \mathbb{R}^{NM}$ such that $x^{\sf T} (\bar{\mathcal{L}} \otimes I_M) x = 0$, 
	it must have the form $\mathbf{1}_N \otimes z$, $z \in \mathbb{R}^M$. For another unit vector $x$, 
	which does not have this form, we know that $x^{\sf T} (\bar{\mathcal{L}} \otimes I_M) x \ge \lambda_{2}(\bar{\mathcal{L}}) \|x_{\perp}\|^2 > 0$, where $x_{\perp}$ is the difference of $x$ minus its projection on the subspace
	$\{\mathbf{1}_N \otimes z : z \in \mathbb{R}^M\}$.
	Now for unit vector $x = \mathbf{1}_N \otimes z$, $z \in \mathbb{R}^M$, 
	\begin{align}\nonumber
	&x^{\sf T} \left( \sum_{t=mh}^{(m+1)h-1} \mathbb{E}\{\bar D_H(t) \bar D_H^{\sf T}(t) |\mathcal{F}(mh-1)\}\right) x\\\nonumber
	&=
	\sum_{t=mh}^{(m+1)h-1}\mathbb{E}\left\{\sum_{j=1}^{N} z^{\sf T} H_{j}^{\sf T}(t) H_{j}(t) z|\mathcal{F}(mh-1)\right\}\\\nonumber
	&=
	z^{\sf T} \left(\sum_{t=mh}^{(m+1)h-1}\mathbb{E}\left\{\sum_{j=1}^{N}H_{j}^{\sf T}(t)H_{j}(t)|\mathcal{F}(mh-1)\right\}\right) z \ge \tilde{\lambda} \|z\|^2 = \tilde{\lambda}/N,
	\end{align}
	where the last inequality follows from \eqref{con_obser}. Denoting 
	\begin{align*}
	\beta := \underset{\|x\| = 1}{\inf_{x \in \mathbb{R}^{NM}}} x^{\sf T} \left(h(\bar{\mathcal{L}} \otimes I_M) +  \sum_{t=mh}^{(m+1)h-1} \mathbb{E}\{\bar D_H(t) \bar D_H^{\sf T}(t) |\mathcal{F}(mh-1)\}\right) x,
	\end{align*}
	since the function in the above definition is positive for every unit vector, and the unit sphere is compact, we know that $\beta > 0$.
\end{proof}
\begin{remark}
	The observability condition (or persistent excitation condition) in \eqref{con_obser} is the possibly mildest, since it is a spatially and temporally  
	collective condition of all sensors and allows  the measurement matrices to be time-varying, stochastic, and non-stationary.  Similar  conditions are given in  \cite{guo1994stability} and \cite{wang2019distributed}   under centralized and distributed settings, respectively. 		
	In the literature,    time-invariant measurement matrices and stationary measurement matrices  are studied in \cite{kar2011convergence,kar2013distributed}  and
	\cite{zhang2012distributed}, respectively.
\end{remark}


\begin{theorem}\label{thm_mean_square_convergence}
	(Mean-square convergence) Under Assumption \ref{asmp_consistency2} (i.a), (ii.a-c), and (iii.a), the estimation error, $e_0(t) = X(t) - \mathbf{1}_N \otimes \theta$, converges to zero in mean square  with a rate $o(\alpha^{2\delta}(t))$, i.e.,
	\begin{align}
	\lim_{t\to \infty} \frac{\mathbb{E}\{\|e_0(t)\|^2\}}{\alpha^{2\delta}(t)} = 0,
	\end{align}
	where $\delta \in [0,1/2)$ satisfies Assumption \ref{asmp_consistency2} (ii.c) and (iii.a).
\end{theorem}

\begin{proof}
	See Section \ref{pf_thm_mean_square_convergence}.
\end{proof}
{
	\begin{remark}
		Under the conditions of Theorem \ref{thm_mean_square_convergence}, a larger $\lambda$ in \eqref{eq_obser_condi} leads to a larger $\delta$ meaning a faster convergence rate. From Proposition \ref{prop_obser_condi}, the value of  $\lambda$ is determined by the  network structure $\mathcal{L}$ and the level of system observability $\sum_{t=mh}^{(m+1)h-1} \mathbb{E}\{\bar D_H(t) \bar D_H^{\sf T}(t) |\mathcal{F}(mh-1)\}$. Thus, in order to obtain a faster convergence rate, the system designer can offline adjust the system and network structure, such as choosing sensors with larger signal-to-noise ratio, new sensor deployment, etc. {Moreover, if we increase the decay speed of triggering threshold $f_i(t)$, then according to (iii) of  Assumption~\ref{asmp_consistency2},  a larger $\delta$ will be available such that the convergence rate of the estimation error increases. 
		}
	\end{remark}
}

In order to analyze   the decay speed of the communication rate over the whole sensor network in Section \ref{sec_sub:commu}, we need to establish the almost-sure convergence of Algorithm \ref{alg:A}. 
As is known, there is a gap  between the almost-sure and
mean-square convergence of a random variable sequence unless
the uniform integrability and certain moment conditions are
satisfied  (\cite{poznyak2009advanced}). 
For example,   if $\sum_{t=0}^{\infty}\mathbb{E}\{\|e_0(t)\|^2\}<\infty$ is satisfied, almost-sure convergence will hold  according to Theorem 6.8 in \cite{poznyak2009advanced}. Although we provide the mean-square convergence rate in Theorem~\ref{thm_mean_square_convergence} with $o(\alpha^{2\delta}(t))$, the condition $\sum_{t=0}^{\infty}\mathbb{E}\{\|e_0(t)\|^2\}<\infty$ is not directly satisfied due to $\delta\in[0,1/2)$. 
In the following theorem, we show that if some slightly stronger conditions on the step size and the triggering threshold are satisfied, Algorithm \ref{alg:A} will have almost-sure convergence with a certain convergence rate.

\begin{theorem}\label{thm_conver_rate}
	(Almost-sure convergence)    Under Assumption \ref{asmp_consistency2}, the estimation error, $e_0(t) = X(t) - \mathbf{1}_N \otimes \theta$, converges to zero almost surely with a rate $o(\alpha^{\delta}(t))$, i.e.,
	\begin{align}
	\mathbb{P} \left\{\lim_{t\to \infty} \frac{e_0(t)}{\alpha^{\delta}(t)} = 0 \right\} = 1,
	\end{align}
	where $\delta \in [0,\frac12)$ satisfies Assumption \ref{asmp_consistency2} (i.b), (ii.c) and (iii.a-b).
\end{theorem}

\begin{proof}
	See Section \ref{pf_thm_conver_rate}.
\end{proof}
\begin{remark}
	Theorems \ref{thm_mean_square_convergence} and \ref{thm_conver_rate} show that the    mean-square and almost-sure convergence rates of Algorithm \ref{alg:A} are $o(\alpha^{2\delta}(t))$ and  $o(\alpha^{\delta}(t))$, respectively. Since $\delta\in[0,1/2)$, the conclusions in Theorems \ref{thm_mean_square_convergence} and \ref{thm_conver_rate}  reduce to the  mean-square and almost-sure convergence (\cite{zhang2012distributed,wang2019distributed}) when $\delta=0$. 
\end{remark}

\subsection{Communication rate}\label{sec_sub:commu}
{After studying the estimation performance of Algorithm~\ref{alg:A}, 
	in this subsection	we  find conditions such that the communication rate of sensors using the algorithm decays to zero with a certain speed while ensuring the estimation error convergence.}
To proceed, we need some extra conditions  on the step size $\alpha_i(t)$ and the triggering threshold $f_i(t)$.

\begin{assumption}\label{ass_triggering}
	There are two  monotonically non-increasing sequences $\{\bar f(t)\}_{t=0}^{\infty}$ and $\{\beta(t)\}_{t=0}^{\infty}$, and a scalar $\mu\in[1/2,1),$ such that: \\
	(i)  $\bar f(t)\leq f_{\min}(t) $ for any $t\in\mathbb{N}$.	\\
	{(ii) $\liminf\limits_{t\rightarrow\infty}\bar f(t+t^{\mu})/\bar f(t)>0$.}\\
	(iii) $	\beta(t)=O(\alpha(t)^{1-2(1-\delta)/\rho}),$ 	where $\delta$ and $\rho$ are given in Theorem \ref{thm_conver_rate} and Assumption~\ref{asmp_consistency2}~(ii.b), respectively.\\
	{	(iv)  $\liminf\limits_{t\rightarrow\infty}\frac{\bar f(t)}{t^{\mu}\beta(t)}>0$.}
\end{assumption}
{
	\begin{remark}\label{rem_mu}
		In (i) and (iii) of Assumption~\ref{ass_triggering},	although monotonicity is required for $\bar f(t)$ and $\beta(t)$, triggering threshold $f_i(t)$ and step size $\alpha_i(t)$ are not necessarily monotonic.  Conditions (ii) and (iv)  are satisfied if $\bar f(t)$ decays slowly, such as the case in Remark~\ref{rem_step}. 	
		The parameter $\mu$ can reflect the decay speed of $\bar f(t)$. One can find a larger $\mu$ if $\bar f(t)$ decays slower.
	\end{remark}	
}
\begin{remark}\label{rem_step}
	The requirements on the step size $\alpha_i(t)$ and the threshold $f_i(t)$ in Assumptions \ref{asmp_consistency2} and \ref{ass_triggering} can be satisfied at the same time. For instance, let $\mu=1/2$, $\alpha(t)=\alpha_i(t)=t^{-1}$ and $\bar f(t)=f_i(t)=t^{-\epsilon_0}$ with $\epsilon_0>\delta\geq 0$,                                                          and $\beta(t)=t^{{-1+2(1-\delta)/\rho}}$.
	Then   the requirements   are satisfied if  $\epsilon_0+2(1-\delta)/\rho\in (0,1/2]$ and $\epsilon_0>\delta\geq 0$.
\end{remark}

\begin{theorem}\label{thm_rate}
	(Communication rate decay)	Under Assumptions \ref{asmp_consistency2} and \ref{ass_triggering},  the communication rate $\lambda_c(t)$ of Algorithm \ref{alg:A} converges to zero almost surely with a rate $o(t^{-\gamma})$ for any {$\gamma\in[0,\frac{2\mu}{2\mu+1})$}, i.e.,
	\begin{align}
	\mathbb{P} \left\{\lim_{t\to \infty} \lambda_c(t)t^{\gamma} = 0 \right\} = 1.
	\end{align}
\end{theorem}
\begin{proof}
	See Section \ref{pf_thm_rate}.
\end{proof}

{{
		\begin{remark}\label{rem_tradeoff} (Tradeoff)
			Given the step size $\alpha_i(t)$, 
			if we reduce the decay speed of the triggering threshold $f_i(t)$, we are able to obtain a larger 
			parameter $\mu$ according to Remark~\ref{rem_mu}. Then according to  Theorem~\ref{thm_rate}, the communication rate $ \lambda_c(t)$ will have a faster decay speed. However, the convergence rate of the estimation error will be reduced since $\delta$ in Theorems~\ref{thm_mean_square_convergence} and \ref{thm_conver_rate} becomes smaller according to (iii) of  Assumption~\ref{asmp_consistency2}. Thus, the decay speed of the triggering threshold  can lead to a tradeoff between the convergence rate of the estimation error and the decay speed of the communication rate.
		\end{remark}
	}
	Theorem~\ref{thm_rate}, together with Remark \ref{rem_step} and Theorems~\ref{thm_mean_square_convergence} and \ref{thm_conver_rate}, shows that the communication rate is decaying to zero with guaranteed convergence of Algorithm \ref{alg:A}, 
	which  means that
	considerable communications between sensors can be effectively reduced comparing to the existing time-triggered approaches \cite{rad2010distributed,kar2011convergence,kar2013distributed,cattivelli2010diffusionLMS,zhang2012distributed,wang2019distributed}.
	
	
	
	
	
	%
	
	\section{Proofs of the main results}\label{sec:proofs}
	{The proofs of the main results in the last section are provided in this section.}
	\subsection{Convergence of linear recursion}\label{subsec_linear_recursion}
	
	To show the convergence of Algorithm~\ref{eq_estimation_compact}, we first study the following linear recursion
	\begin{align}\label{eq_linear_recursion}
	e(t+1) = e(t) + \alpha(t) (Q(t) + \Delta(t)) e(t) + \alpha(t)(\varepsilon^{\prime}(t) + \varepsilon^{\prime\prime}(t)),
	\end{align}
	where $e(t), \varepsilon^{\prime}(t), \varepsilon^{\prime\prime}(t) \in \mathbb{R}^{q}$, $Q(t), \Delta(t) \in \mathbb{R}^{q\times q}$, and $\alpha(t) \in \mathbb{R}$ is the step size, $t\in \mathbb{N}$.
	Some assumptions are introduced below.
	\begin{assumption}\label{asmp_linear_recursion}~\\
		(i)  $\{\alpha(t)\}$ satisfies that $\alpha(t) > 0$, $\alpha(t) \to 0$,  $\sum_{t=1}^{\infty} \alpha(t) = \infty$, and
		\begin{equation}\label{eq_asmp_linear_recursion_stepsize_rate}
		\frac{1}{\alpha(t+1)} - \frac{1}{\alpha(t)} \to \alpha_0 \ge 0.
		\end{equation} 
		(ii.a) $\{Q(t), t\in \mathbb{N}\}$ is a sequence of random matrices, and there exists a constant $\pi_1 \in \mathbb{R}^+$ such that \begin{align}\label{eq_linear_recursion_boundedQ}
		\sup_{t\in\mathbb{N}} \|Q(t)\| \le \pi_1, \text{ a.s.}
		\end{align}
		In addition, there exist $h\in\mathbb{N}^+$, $\lambda \in \mathbb{R}^+$, such that for any $m\in\mathbb{N}$,
		\begin{align}\label{eq_linear_recursion_obser_condi}
		\lambda_{\max}\left[\sum_{t=mh}^{(m+1)h-1}\mathbb{E}\left\{Q(t) + Q^{\sf T}(t)|\mathcal{F}(mh-1)\right\}\right]\le -\lambda, \text{ a.s.,}
		\end{align}
		where $\mathcal{F}(t) := \sigma(Q(s), \Delta(s), \varepsilon^{\prime}(s), 0\le s\le t)$ and $\mathcal{F}(-1) = \{\emptyset, \Omega\}$.\\
		(ii.b) $\{\Delta(t), t\in \mathbb{N}\}$ is a sequence of random matrices such that for $t \in \mathbb{N}$
		\begin{align*}
		\|\Delta(t)\| \le g_1(t), \text{ a.s.,}
		\end{align*}
		where $g_1(t): \mathbb{R}^+ \cup \{0\} \to \mathbb{R}^+$ is a measurable function satisfying $g_1(t) \to 0$ as $t \to \infty$ and $\sup_{t\ge 0} g_1(t) <~\infty$.\\
		(iii.a) $\{\varepsilon^{\prime}(t), \mathcal{F}(t)\}$ is a martingale difference sequence, i.e., $\mathbb{E}\{\varepsilon^{\prime}(t)|\mathcal{F}(t-1)\} = 0$ for all $t \in \mathbb{N}$, and $\sup_{t\in\mathbb{N}} \mathbb{E}\{\|\alpha^{\delta}(t)\varepsilon^{\prime}(t)\|^2|\mathcal{F}(t-1)\} \le c_\varepsilon$ a.s. with $c_\varepsilon$ a positive constant and some $\delta \in [0, 1/2)$.\\
		(iii.b) $\varepsilon^{\prime\prime}(t) \in \mathcal{F}(t)$, and $\|\varepsilon^{\prime\prime}(t)\| \le g_2(t)$ a.s., $t \in \mathbb{N}$, where $g_2(t): \mathbb{R}^+ \cup \{0\} \to \mathbb{R}^+$ is a measurable function satisfying $g_2(t) \to 0$ as $t \to \infty$ and $\sup_{t\ge 0} g_2(t) < \infty$.
	\end{assumption}
	
	{\begin{remark}\label{rem_asmp_linear_recursion}
			In the above assumption, the first three conditions of (i) are standard for step sizes, and~\eqref{eq_asmp_linear_recursion_stepsize_rate} is a mild condition used in the characterization of the convergence rate of an algorithm \cite{chen2002stochastic}. The boundedness of $Q(t)$ in~\eqref{eq_linear_recursion_boundedQ} is assumed for simplicity but could be extended to uniform conditional boundedness with respect to $\mathcal{F}(t-1)$. Similarly, The dominance condition in (ii.b) could also be extended to certain conditional dominance by $g_1(t)$ with respect to $\mathcal{F}(t-1)$. Condition~\eqref{eq_linear_recursion_obser_condi} is a counterpart of the persistent excitation condition for centralized algorithms \cite{guo1990estimating}. (iii.a) is a common noise assumption and guarantees that $\sum_{t=1}^{\infty} \alpha(t) \varepsilon^{\prime}(t) = \sum_{t=1}^{\infty} \alpha^{1-\delta}(t) (\alpha^{\delta}(t) \varepsilon^{\prime}(t)) < \infty$ a.s. when $\sum_{t=1}^{\infty} \alpha^{2(1-\delta)}(t) < \infty$. In addition, (iii.b) ensures $\varepsilon^{\prime\prime}(t) \to 0$. These two conditions, $\sum_{t=1}^{\infty} \alpha(t) \varepsilon^{\prime}(t) < \infty$ and $\varepsilon^{\prime\prime}(t) \to 0$, are standard for a.s. convergence of stochastic approximation algorithms \cite{chen2002stochastic}. Extensions of (iii.a) and (iii.b) are left to future work.
	\end{remark}}
	
	\begin{theorem}\label{thm_linear_recursion}
		Under Assumption~\ref{asmp_linear_recursion}, $e(t)$ in \eqref{eq_linear_recursion}, starting with any fixed initial condition, converges in mean square, i.e.,
		\begin{align*}
		\lim_{t\to\infty}\mathbb{E}\{\|e(t)\|^2\} = 0.
		\end{align*}
		In addition, if the conditions below hold\\
		(a) $\sum_{t=1}^{\infty} \alpha(t)^{2(1-\delta)} < \infty$, where $\delta$ is defined in Assumption~\ref{asmp_linear_recursion}~(iii.a),\\
		(b) $\sum_{t=1}^{\infty} \alpha(t) g_2(t) < \infty$, \\
		(c) $\varepsilon^{\prime}(t)$ can be written as $\varepsilon^{\prime}(t) = u(t) w(t)$, where $u(t) \in \mathbb{R}^{q\times q}$ and $w(t) \in \mathbb{R}^q$, such that $\{w(t), t\in\mathbb{N}\}$ is independent of $\{Q(t), \Delta(t), u(t), t\in\mathbb{N}\}$ and $\{w(t),\mathcal{F}(t)\}$ is a martingale difference sequence, \\
		then $e(t)$ starting with any fixed initial condition converges to zero a.s., i.e.,
		\begin{align*}
		\mathbb{P}\left\{ \lim_{t\to\infty} e(t) = 0 \right\} = 1.
		\end{align*}
	\end{theorem}
	
	{\begin{remark}
			As mentioned in Remark~\ref{rem_asmp_linear_recursion}, the additional condition (a) is to ensure $\sum_{t=1}^{\infty} \alpha(t) \varepsilon^{\prime}(t) < \infty$ a.s. Conditions~(b) and~(c) are introduced to deal with the possible dependence of measurement noise, $\epsilon^\prime(t)$ and $\epsilon^{\prime\prime}(t)$, on $Q(t)$ and $\Delta(t)$. This does not exist in classic linear recursion, where $Q(t)$ and $\Delta(t)$ are both deterministic \cite{chen2002stochastic}. Possible extensions of these two conditions will be investigated in the future.
	\end{remark}}
	
Denote $\Phi(t,s) := \Big(I + \alpha(t)\big(Q(t) + \Delta(t)\big)\Big) \Big(I + \alpha(t-1)\big(Q(t-1) + \Delta(t-1)\big)\Big) \cdots \Big(I + \alpha(s)\big(Q(s) + \Delta(s)\big)\Big)$, $t\ge s \ge 0$, and $\Phi(t,s) = I$, $0 \le t < s$. To prove Theorem~\ref{thm_linear_recursion}, we begin with the following lemma.

\begin{lemma}\label{lem_linear_recursion}
	Under Assumption~\ref{asmp_linear_recursion}~(i), (ii.a), and~(ii.b), there exists a positive integer $\tilde{m}$ such that for $t \ge \tilde{m}h$ and $s \ge 0$, 
	\begin{align*}
	\mathbb{E}\left\{\Phi^{\sf T}(t,s)\Phi(t,s) | \mathcal{F}(s-1)\right\} \le c_2 \exp\left(- c_1 \sum_{k=s}^{t} \alpha(k)\right) I,~ \text{ a.s.,}
	\end{align*}
	where $c_1$ and $c_2$ are positive constants.
\end{lemma}

\begin{proof}
	First of all, we bound $\Phi^{\sf T}(t,s)\Phi(t,s)$ in one period of length $h$, i.e., for $m \in \mathbb{N}$,
	\begin{align}\nonumber
	&\Phi^{\sf T}\big( (m+1)h-1,mh \big) \Phi\big( (m+1)h-1,mh \big)\\\nonumber
	&= \Big(I + \alpha(mh) \big(Q(mh) + \Delta(mh)\big)\Big)^{\sf T} \cdots \Big(I + \alpha\big((m+1)h-1\big) \\\nonumber
	&\quad \big(Q\big((m+1)h-1\big)  + \Delta\big((m+1)h-1\big)\big)\Big)^{\sf T}  \Big(I + \alpha\big((m+1)h-1\big) \\\nonumber
	&\quad \big(Q\big((m+1)h-1\big)  + \Delta\big((m+1)h-1\big)\big)\Big) \cdots \Big(I + \alpha(mh) \big(Q(mh) + \Delta(mh)\big)\Big)\\\label{eq_lem_linear_recursion_1period}
	&= 
	I + \sum_{t=mh}^{(m+1)h-1} \alpha(t) \big(Q(t) + Q^{\sf T}(t) +  \Delta(t) + \Delta^{\sf T}(t) \big) + \sum_{k=2}^{2h} M_{k}(m),
	\end{align}
	where $M_k(m)$, $2\le k \le 2h$, is the $k$-th order term in the binomial expansion of $\Phi^{\sf T}\big((m+1)h-1,mh\big) \Phi\big((m+1)h-1,mh\big)$.
	
	Note that from Assumption~\ref{asmp_linear_recursion}~(i), we know that for fixed $k \ge 1$,
	\begin{align*}
	\frac{1}{\alpha(t+k)} - \frac{1}{\alpha(t)} \to k\alpha_0.
	\end{align*}
	So for fixed $k \ge 1$ and $0\le i \not= j\le k$, it holds that
	\begin{align}\nonumber
	\frac{\alpha(t+i)}{\alpha(t+j)} &= \frac{\alpha(t+i)}{\alpha(t+j)} - 1 + 1\\\nonumber
	&= 1 + \alpha(t+i)\left( \frac{1}{\alpha(t+j)} - \frac{1}{\alpha(t+i)} \right)\\\nonumber
	&= 1 + \alpha(t+i)\big((j-i)\alpha_0 + o(1) \big)\\\label{eq_lem_linear_recursion_stepsize1}
	&= 1 + O\big(\alpha(t+i)\big),
	\end{align}
	and
	\begin{align}\nonumber
	\alpha(t+i) - \alpha(t+j) &= \alpha(t+i) \alpha(t+j) \left(\frac{1}{\alpha(t+j)} - \frac{1}{\alpha(t+i)} \right)\\\nonumber
	&= \alpha(t+i) \alpha(t+j) \big((j-i)\alpha_0 + o(1) \big)\\\label{eq_lem_linear_recursion_stepsize2}
	&= o\big(\alpha(t+j)\big).
	\end{align}
	Hence, for all $0 \le i \le h-1$,
	\begin{align}\label{use_eq_lem_linear_recursion_stepsize1}
	\eqref{eq_lem_linear_recursion_1period} &= I + \sum_{t=mh}^{(m+1)h-1} \alpha(t) \big(Q(t) + Q^{\sf T}(t) +  \Delta(t) + \Delta^{\sf T}(t) \big) \\\nonumber
	&\quad + o\big(\alpha(mh+i)\big)\\\nonumber
	&= 
	I + \alpha(mh+i) \sum_{t=mh}^{(m+1)h-1} \big(Q(t) + Q^{\sf T}(t)\big) \\\nonumber
	&\quad + \alpha(mh+i) \sum_{t=mh}^{(m+1)h-1} \big(\Delta(t) + \Delta^{\sf T}(t) \big) \\\nonumber
	&\quad +
	\sum_{t=mh}^{(m+1)h-1} \big(\alpha(t) - \alpha(mh+i)\big) \big(Q(t) + Q^{\sf T}(t) +  \Delta(t) + \Delta^{\sf T}(t) \big) \\\nonumber
	&\quad +o\big(\alpha(mh+i)\big)\\\label{use_eq_lem_linear_recursion_stepsize2}
	&=
	I + \alpha(mh+i) \sum_{t=mh}^{(m+1)h-1}  \big(Q(t) + Q^{\sf T}(t)\big) + o\big(\alpha(mh+i)\big),
	\end{align}
	where \eqref{use_eq_lem_linear_recursion_stepsize1} follows from \eqref{eq_lem_linear_recursion_stepsize1}, the definition of $M_k(m)$, $2\le k \le 2h$, and Assumptions~\ref{asmp_linear_recursion}~(ii.a)-(ii.b), and $o(\alpha(mh+i)) \to 0$ a.s. as $m \to \infty$ for fixed $0 \le i \le h-1$.~\eqref{use_eq_lem_linear_recursion_stepsize2} is obtained from \eqref{eq_linear_recursion_boundedQ}, \eqref{eq_lem_linear_recursion_stepsize2}, and Assumption~\ref{asmp_linear_recursion}~(ii.b). Taking conditional expectation of \eqref{eq_lem_linear_recursion_1period} with respect to $\mathcal{F}(mh-1)$, we have that
	\begin{align}\nonumber
	&\mathbb{E}\left\{\Phi^{\sf T}\big( (m+1)h-1,mh \big) \Phi\big( (m+1)h-1,mh \big) \big| \mathcal{F}(mh-1) \right\}\\\nonumber
	&=
	\mathbb{E}\left\{ I + \alpha(mh+i) \sum_{t=mh}^{(m+1)h-1}  \big(Q(t) + Q^{\sf T}(t)\big) + o\big(\alpha(mh+i)\big) \big| \mathcal{F}(mh-1) \right\}\\\nonumber
	&=
	I + \alpha(mh+i) \sum_{t=mh}^{(m+1)h-1} \mathbb{E}\left\{ Q(t) + Q^{\sf T}(t) \big| \mathcal{F}(mh-1) \right\}  + o\big(\alpha(mh+i)\big) \\\label{use_asmp_linear_recursion_iia_1}
	&\le
	I - \alpha(mh+i) \lambda I  + o\big(\alpha(mh+i)\big)\\\nonumber
	&=
	I - \frac{\lambda}{h} \left(\sum_{t=mh}^{(m+1)h-1} \alpha(t)\right) I + \frac{\lambda}{h} \left(\sum_{t=mh}^{(m+1)h-1} \big(\alpha(t) - \alpha(mh+i) \big)\right) I \\\nonumber
	&\quad + o\big(\alpha(mh+i)\big)\\\label{use_eq_lem_linear_recursion_stepsize2_b}
	&=
	\left(1 - \frac{\lambda}{h} \sum_{t=mh}^{(m+1)h-1} \alpha(t) + o\big(\alpha(mh+i)\big) \right) I, 
	\end{align}
	where \eqref{use_asmp_linear_recursion_iia_1} is due to \eqref{eq_linear_recursion_obser_condi} and \eqref{use_eq_lem_linear_recursion_stepsize2_b} follows from \eqref{eq_lem_linear_recursion_stepsize2}. Under Assumptions~\ref{asmp_linear_recursion}~(ii.a)-(ii.b), $o\big(\alpha(mh+i)\big)$ is bounded by a deterministic function tending to zero, so for large enough $m$, say $m \ge \tilde{m}$,
	\begin{align}\label{use_eq_lem_linear_recursion_stepsize2_c}
	\eqref{use_eq_lem_linear_recursion_stepsize2_b} &\le
	\left(1 - \frac{\lambda}{2h} \sum_{t=mh}^{(m+1)h-1} \alpha(t) \right) I \\\label{use_eq_elementary1}
	&\le
	\exp \left( - c_1 \sum_{t=mh}^{(m+1)h-1} \alpha(t) \right) I,
	\end{align}
	where~\eqref{use_eq_lem_linear_recursion_stepsize2_c} holds for $m \ge \tilde{m}$ by noticing that $\alpha(t) > 0$, and~\eqref{use_eq_elementary1} with $c_1:= \lambda/2h$ is obtained from $1 - x \le e^{-x}$ for all $x \in \mathbb{R}$.
	
	Note that it follows from \eqref{eq_linear_recursion_boundedQ} that $\Phi^{\sf T}(mh+l,mh)\Phi(mh+l,mh)$ is uniformly bounded a.s. for all $m \in \mathbb{N}$ and $0\le l < h$. Moreover, from $\alpha(t) \to 0$, $\exp\left(c_1 \sum_{k=mh}^{mh+l} \alpha(k)\right) \to 1$ as $m \to \infty$, so $\Phi^{\sf T}(mh+l,mh)\Phi(mh+l,mh)$ can be uniformly bounded a.s. by $c_0 \exp\left(- c_1 \sum_{k=mh}^{mh+l} \alpha(k)\right)$ for all $m \in \mathbb{N}$ and $0\le l < h$, where $c_0$ is a positive constant.
	
	Now, for $\Phi^{\sf T}(t,s)\Phi(t,s)$ starting with $s = \tilde{m}h$, it holds for $t = mh + l$ with $m \ge \tilde{m}$ and $0\le l \le h-1$ that
	\begin{align}\nonumber
	&\mathbb{E}\left\{\Phi^{\sf T}(t,\tilde{m}h)\Phi(t,\tilde{m}h) | \mathcal{F}(\tilde{m}h-1) \right\}\\\nonumber
	&=
	\mathbb{E}\left\{\Phi^{\sf T}(mh-1,\tilde{m}h) \big(\Phi^{\sf T}(t,mh) \Phi(t,mh) \big) \Phi(mh-1,\tilde{m}h) | \mathcal{F}(\tilde{m}h-1) \right\}\\\nonumber
	&\le
	c_0 \exp\left(- c_1 \sum_{k=mh}^{t} \alpha(k)\right) \mathbb{E}\left\{\Phi^{\sf T}(mh-1,\tilde{m}h) \Phi(mh-1,\tilde{m}h) | \mathcal{F}(\tilde{m}h-1) \right\}\\\nonumber
	&=
	c_0 \exp\left(- c_1 \sum_{k=mh}^{t} \alpha(k)\right) \mathbb{E}\big\{\Phi^{\sf T}\big((m-1)h-1,\tilde{m}h\big) \\\nonumber
	&\quad \Phi^{\sf T}\big(mh-1,(m-1)h\big)   \Phi\big(mh-1,(m-1)h\big) \\\nonumber 
	&\quad \Phi\big((m-1)h-1,\tilde{m}h\big) \big| \mathcal{F}(\tilde{m}h-1) \big\} \\\label{eq_ms_converg_Phit_2}
	&= 
	c_0 \exp\left(- c_1 \sum_{k=mh}^{t} \alpha(k)\right) \mathbb{E}\Big\{ \mathbb{E}\big\{ \Phi^{\sf T}\big((m-1)h-1,\tilde{m}h\big) \\\nonumber
	&\quad \Phi^{\sf T}\big(mh-1,(m-1)h\big) \Phi\big(mh-1,(m-1)h\big) \\\nonumber
	&\quad  \Phi\big((m-1)h-1,\tilde{m}h\big) \big| \mathcal{F}\big((m-1)h-1\big) \big\} \Big| \mathcal{F}(\tilde{m}h-1) \Big\} \\\nonumber
	&= 
	c_0 \exp\left(- c_1 \sum_{k=mh}^{t} \alpha(k)\right) \mathbb{E}\Big\{ \Phi^{\sf T}\big((m-1)h-1,\tilde{m}h\big) \\\nonumber
	&\quad  \mathbb{E}\big\{ \Phi^{\sf T}\big(mh-1,(m-1)h\big) \Phi\big(mh-1,(m-1)h\big) \big| \mathcal{F}\big((m-1)h-1\big) \big\} \\\nonumber
	&\quad  \Phi\big((m-1)h-1,\tilde{m}h\big) \Big| \mathcal{F}(\tilde{m}h-1) \Big\} \\\label{eq_use_use_eq_elementary1_a}
	&\le 
	c_0 \exp\left(- c_1 \sum_{k=(m-1)h}^{t} \alpha(k)\right) \mathbb{E}\big\{ \Phi^{\sf T}\big((m-1)h-1,\tilde{m}h\big) \\\nonumber
	&\quad \Phi\big((m-1)h-1,\tilde{m}h\big) \big| \mathcal{F}(\tilde{m}h-1) \big\} \\\nonumber
	&\le 
	\dots\\\label{eq_repeatly_use_formertwo}
	&\le 
	c_0 \exp\left(- c_1 \sum_{k=\tilde{m}h}^{t} \alpha(k)\right) I,
	\end{align}
	where \eqref{eq_ms_converg_Phit_2} is implied by $\mathbb{E}\{X|\mathcal{F}_1\} = \mathbb{E}\{\mathbb{E}\{X|\mathcal{F}_2\}|\mathcal{F}_1\}$ for $\mathcal{F}_1 \subset \mathcal{F}_2$, \eqref{eq_use_use_eq_elementary1_a} follows from \eqref{use_eq_elementary1}, and \eqref{eq_repeatly_use_formertwo} is obtained by repeatly using the two former arguments.
	
	Finally, from \eqref{eq_linear_recursion_boundedQ} we know that $\Phi^{\sf T}(\tilde{m}h-1, s)\Phi(\tilde{m}h-1,s)$ can be uniformly bounded a.s. for all $0 \le s \le \tilde{m}h-1$ since $\tilde{m}$ is fixed. Hence, for $t \ge \tilde{m}h$ and $s \ge 0$ we have 
	\begin{align}\nonumber
	&\mathbb{E}\left\{\Phi^{\sf T}(t,s)\Phi(t,s) | \mathcal{F}(s-1)\right\} \\\nonumber
	&=
	\mathbb{E}\left\{\Phi^{\sf T}(\tilde{m}h-1,s) \big(\Phi^{\sf T}(t,\tilde{m}h)  \Phi(t,\tilde{m}h) \big) \Phi(\tilde{m}h-1,s) \big| \mathcal{F}(s-1) \right\} \\\label{eq_use_eq_repeatly_use_formertwo}
	&\le 
	c_0 \exp\left(- c_1 \sum_{k=\tilde{m}h}^{t} \alpha(k)\right)  \mathbb{E}\left\{ \Phi(\tilde{m}h-1,s)^{\sf T} \Phi(\tilde{m}h-1,s)  | \mathcal{F}(s-1) \right\} \\\nonumber
	&\le c_2 \exp\left(- c_1 \sum_{k=s}^{t} \alpha(k)\right) I,
	\end{align}
	where \eqref{eq_use_eq_repeatly_use_formertwo} follows from \eqref{eq_repeatly_use_formertwo}. 
\end{proof}

\noindent\emph{\textbf{Proof of Theorem~\ref{thm_linear_recursion}}:}

From \eqref{eq_linear_recursion} and the definition of $\Phi(t,s)$, we have that
\begin{align*}
e(t+1) &= \big(I + \alpha(t) (Q(t) + \Delta(t)) \big) e(t) + \alpha(t)(\varepsilon^{\prime}(t) + \varepsilon^{\prime\prime}(t) )\\
&= 
\Phi(t,0) e(0) + \sum_{k=0}^t \alpha(k) \Phi(t, k+1) (\varepsilon^{\prime}(k) + \varepsilon^{\prime\prime}(k)).
\end{align*}
First we prove the mean-square convergence part. Write
\begin{align}\nonumber
&\mathbb{E}\{\|e(t+1)\|^2\} = \mathbb{E}\{e^{\sf T}(t+1) e(t+1)\}\\\nonumber
&= 
\mathbb{E}\left\{\left(\Phi(t,0) e(0) + \sum_{k=0}^t \alpha(k) \Phi(t, k+1) (\varepsilon^{\prime}(k) + \varepsilon^{\prime\prime}(k)) \right)^{\sf T} \right.\\\nonumber
&\quad \left. \left( \Phi(t,0) e(0) + \sum_{k=0}^t \alpha(k) \Phi(t, k+1) (\varepsilon^{\prime}(k) + \varepsilon^{\prime\prime}(k)) \right) \right\}\\\nonumber
&\le 
3 \bigg(\mathbb{E}\{\left\| \Phi(t,0) e(0) \right\|^2 \} + \mathbb{E} \left\{ \left\| \sum_{k=0}^t \alpha(k) \Phi(t, k+1) \varepsilon^{\prime}(k) \right\|^2 \right\}\nonumber\\
&\quad  + \mathbb{E} \left\{ \left\| \sum_{k=0}^t \alpha(k) \Phi(t, k+1) \varepsilon^{\prime\prime}(k) \right\|^2 \right\} \bigg)\\\nonumber
&:= 3 ((I) + (II) + (III)).
\end{align}
For the first part,
\begin{align}\nonumber
(I) &= \mathbb{E}\{\left( \Phi(t,0) e(0) \right)^{\sf T}\left( \Phi(t,0) e(0) \right) \}\\\nonumber
&= \mathbb{E}\{e^{\sf T}(0) \Phi^{\sf T}(t,0)\Phi(t,0) e(0)\}\\\nonumber
&\le c_2 \exp\left(- c_1 \sum_{k=0}^{t} \alpha(k)\right) \mathbb{E}\{\|e(0)\|^2\} \to 0,
\end{align}
as $t \to \infty$, from Lemma~\ref{lem_linear_recursion} and Assumption~\ref{asmp_linear_recursion}~(i). 

For the second part, 
\begin{align}\nonumber
(II) &= \mathbb{E} \left\{ \left( \sum_{k=0}^t \alpha(k) \Phi(t, k+1) \varepsilon^{\prime}(k) \right)^{\sf T} \left( \sum_{k=0}^t \alpha(k) \Phi(t, k+1) \varepsilon^{\prime}(k) \right) \right\}\\\nonumber
&= \mathbb{E} \left\{ \sum_{k=0}^t \left(\alpha(k) \Phi(t, k+1) \varepsilon^{\prime}(k) \right)^{\sf T} \left( \alpha(k) \Phi(t, k+1) \varepsilon^{\prime}(k) \right) \right\},
\end{align}
because Assumption \ref{asmp_linear_recursion} (iii.a) yields that for $i > j \ge 0$
\begin{align*}
\mathbb{E}\{\varepsilon^{\prime}(i) \varepsilon^{\prime}(j)^{\sf T}\} =\mathbb{E}\{\mathbb{E}\{\varepsilon^{\prime}(i) \varepsilon^{\prime}(j)^{\sf T}|\mathcal{F}(j)\}\} =\mathbb{E}\{\mathbb{E}\{\varepsilon^{\prime}(i)|\mathcal{F}(j)\}\varepsilon^{\prime}(j)^{\sf T}\} = 0.
\end{align*}
Now
\begin{align}\nonumber
(II)  
&= 
\mathbb{E} \left\{ \sum_{k=0}^t \left(\alpha(k) \Phi(t, k+1) \varepsilon^{\prime}(k) \right)^{\sf T} \left( \alpha(k) \Phi(t, k+1) \varepsilon^{\prime}(k) \right) \right\}\\\nonumber
&=
\sum_{k=0}^t \alpha^{2}(k) \mathbb{E} \left\{ \varepsilon^{\prime}(k)^{\sf T} \Phi^{\sf T}(t, k+1) \Phi(t, k+1) \varepsilon^{\prime}(k) \right\}\\\nonumber
&=
\sum_{k=0}^t \alpha^{2}(k) \mathbb{E} \left\{ \mathbb{E} \left\{ \varepsilon^{\prime}(k)^{\sf T} \Phi^{\sf T}(t, k+1)\Phi(t, k+1) \varepsilon^{\prime}(k) | \mathcal{F}(k) \right\} \right\}\\\nonumber
&=
\sum_{k=0}^t \alpha^{2}(k) \mathbb{E} \left\{ \varepsilon^{\prime}(k)^{\sf T} \mathbb{E} \left\{  \Phi^{\sf T}(t, k+1)\Phi(t, k+1)  | \mathcal{F}(k) \right\} \varepsilon^{\prime}(k) \right\}\\\label{eq_ms_converg_II_1}
&\le
\sum_{k=0}^t \alpha^{2(1-\delta)}(k)~ c_2 \exp\left(- c_1 \sum_{i=k+1}^{t} \alpha(i)\right) \mathbb{E} \left\{ \|\alpha^{\delta}(k) \varepsilon^{\prime}(k) \|^2 \right\}\\\label{eq_ms_converg_II_3}
&\le 
c_2 c_{\varepsilon} \sum_{k=0}^t \alpha^{2(1-\delta)}(k) ~ \exp\left(- c_1 \sum_{i=k+1}^{t} \alpha(i)\right),
\end{align}
where \eqref{eq_ms_converg_II_1} follows from Lemma \ref{lem_linear_recursion}, and \eqref{eq_ms_converg_II_3} from Assumption \ref{asmp_linear_recursion} (iii.a). 

Note that from Assumption \ref{asmp_linear_recursion} (i), $\alpha(t) \to 0$, so for any $\varepsilon > 0$ and $\delta \in [0, \frac12)$, there exists $t_1$ such that $\alpha^{1-2\delta}(t) < \varepsilon$ for all $t > t_1$. In addition, for large enough $t_2$ and $t > t_2$, we have that 
\begin{align*}
\alpha(t) &\le \alpha(t) + (\alpha(t) - c_1 \alpha^2(t))\\
&= 2 \left(\alpha(t) - \frac{c_1 \alpha^2(t)}{2}\right)\\
&\le \frac{2}{c_1} (1 - \exp(-c_1 \alpha(t))),
\end{align*}
where the last inequality follows from $e^{-x} \le 1-x+x^2/2$, $\forall x \ge 0$. Consequently, we have for $t > t_0 := \max\{t_1, t_2\}$
\begin{align}\label{eq_sum_of_alpha_times_exponential}
&\sum_{k=0}^t \alpha^{2(1-\delta)}(k)  \exp\left(- c_1 \sum_{i=k+1}^{t} \alpha(i)\right)\\\nonumber
&=
\sum_{k=0}^{t_0} \alpha^{2(1-\delta)}(k) \exp\left(- c_1 \sum_{i=k+1}^{t} \alpha(i)\right) + \sum_{k=t_0+1}^t \alpha^{2(1-\delta)}(k)  \exp\left(- c_1 \sum_{i=k+1}^{t} \alpha(i)\right)\\\nonumber
&\le \sum_{k=0}^{t_0} \alpha^{2(1-\delta)}(k)  \exp\left(- c_1 \sum_{i=k+1}^{t} \alpha(i)\right) + \varepsilon \sum_{k=t_0+1}^t \alpha(k) \exp\left(- c_1 \sum_{i=k+1}^{t} \alpha(i)\right).
\end{align}
The first term above converges to zero since $t_0$ is fixed and $\sum_{t=1}^{\infty} \alpha(t) = \infty$ in Assumption~\ref{asmp_linear_recursion}, and the second term with $t > t_0$ is bounded by
\begin{align}\nonumber
&\varepsilon \sum_{k=t_0+1}^t \alpha(k) \exp\left(- c_1 \sum_{i=k+1}^{t} \alpha(i)\right)\\\nonumber
&\le 
\frac{2 \varepsilon}{c_1} \sum_{k=t_0+1}^t (1 - \exp(-c_1 \alpha(k))) \exp\left(- c_1 \sum_{i=k+1}^{t} \alpha(i)\right)\\\nonumber
&\le 
\frac{2 \varepsilon}{c_1} \sum_{k=t_0+1}^t  \left( \exp\left(- c_1 \sum_{i=k+1}^{t} \alpha(i)\right) - \exp\left(- c_1 \sum_{i=k}^{t} \alpha(i)\right) \right)\\\nonumber
&\le \frac{2 \varepsilon}{c_1}.
\end{align}
The arbitrariness of $\varepsilon$ yields \eqref{eq_sum_of_alpha_times_exponential} tends to zero as $t \to \infty$, which further leads to $(II) \to 0$.

Finally, for the third part,
\begin{align}\nonumber
(III) &= 
\mathbb{E} \left\{ \left(\sum_{k=0}^t \alpha(k) \Phi(t, k+1)  \varepsilon^{\prime\prime}(k) \right)^{\sf T} \left(\sum_{k=0}^t \alpha(k) \Phi(t, k+1) \varepsilon^{\prime\prime}(k) \right) \right\}\\\nonumber
&= 
\mathbb{E} \left\{ \left| \sum_{k=0}^t \sum_{l=0}^t\left( \alpha(k) \Phi(t, k+1) \varepsilon^{\prime\prime}(k) \right)^{\sf T} \left(  \alpha(l) \Phi(t, l+1) \varepsilon^{\prime\prime}(l) \right) \right| \right\}\\\nonumber
&\le 
\sum_{k=0}^t \sum_{l=0}^t \alpha(k) \alpha(l) \mathbb{E} \left\{ \left| \left( \Phi(t, k+1) \varepsilon^{\prime\prime}(k) \right)^{\sf T} \left(   \Phi(t, l+1) \varepsilon^{\prime\prime}(l) \right) \right| \right\}\\\nonumber
&\le 
\sum_{k=0}^t \sum_{l=0}^t \alpha(k) \alpha(l) \mathbb{E} \left\{ \left\|   \Phi(t, k+1) \varepsilon^{\prime\prime}(k) \right\| \left\|   \Phi(t, l+1) \varepsilon^{\prime\prime}(l) \right\| \right\}\\\label{eq_ms_converg_holder}
&\le 
\sum_{k=0}^t \sum_{l=0}^t \alpha(k) \alpha(l) \left(\mathbb{E} \left\{ \left\|   \Phi(t, k+1) \varepsilon^{\prime\prime}(k) \right\|^2 \right\}\right)^{\frac12} \left( \mathbb{E}\left\{ \left\|   \Phi(t, l+1) \varepsilon^{\prime\prime}(l) \right\|^2 \right\} \right)^{\frac12}\\\label{eq_ms_converg_III_1}
&=
\left( \sum_{k=0}^t \alpha(k) \left(\mathbb{E} \left\{ \left\|   \Phi(t, k+1) \varepsilon^{\prime\prime}(k) \right\|^2 \right\}\right)^{\frac12}
\right)^2,
\end{align}
where \eqref{eq_ms_converg_holder} follows from H{\"o}lder inequality.

Note that
\begin{align}\nonumber
&\sum_{k=0}^t \alpha(k) \left(\mathbb{E} \left\{ \left\|   \Phi(t, k+1) \varepsilon^{\prime\prime}(k) \right\|^2 \right\}\right)^{\frac12}\\\nonumber
&=
\sum_{k=0}^t \alpha(k) \left(\mathbb{E} \left\{ \varepsilon^{\prime\prime}(k)^{\sf T} \Phi^{\sf T}(t, k+1) \Phi(t, k+1) \varepsilon^{\prime\prime}(k) \right\}\right)^{\frac12} \\\label{eq_ms_converg_lem4_3}
&\le 
\sum_{k=0}^t \alpha(k) \left(c_2 \exp\left(- \frac12 c_1 \sum_{i=k+1}^{t} \alpha(i)\right) (\mathbb{E} \{\|\varepsilon^{\prime\prime}(k)\|^2\})^{\frac12} \right)\\\label{eq_use_asmp_linear_recursion_iiib}
&\le 
c_2 \sum_{k=0}^t \alpha(k) g_2(k) \exp\left(- \frac12 c_1 \sum_{i=k+1}^{t} \alpha(i)\right), 
\end{align}
where \eqref{eq_ms_converg_lem4_3} is implied by Lemma \ref{lem_linear_recursion}, and \eqref{eq_use_asmp_linear_recursion_iiib} is due to Assumption~\ref{asmp_linear_recursion}~(iii.b). Since $g_2(t) \to 0$ as $t\to \infty$, following the same argument for proving $(II) \to 0$, we obtain that $(III)$ tends to zero. In this way, we verify that $\mathbb{E}\{\|e(t)\|^2\} \to 0$ as $t \to \infty$.

To show the a.s. convergence of $e(t)$, we first study the convergence of sequence $\{\|e(mh)\|^2, m\in\mathbb{N}\}$. From \eqref{eq_linear_recursion} and the definition of $\Phi(t,s)$, we have that for $m \in \mathbb{N}$,
\begin{align*}
e\big((m+1)h\big) =& \Phi\big((m+1)h-1,mh\big) e(mh) \\
&+ \sum_{k=mh}^{(m+1)h-1} \alpha(k) \Phi\big((m+1)h-1, k+1\big) (\varepsilon^{\prime}(k) + \varepsilon^{\prime\prime}(k)),
\end{align*}
so
\begin{align}\nonumber
& \mathbb{E} \left\{\left\|e\big((m+1)h\big)\right\|^2 \big| \mathcal{F}(mh-1) \right\}\\\nonumber
&=
\mathbb{E} \left\{ \left\| \Phi\big((m+1)h-1,mh\big) e(mh) \right\|^2 \big| \mathcal{F}(mh-1) \right\} \\\nonumber
&
+ \mathbb{E} \Bigg\{ \bigg\| \sum_{k=mh}^{(m+1)h-1} \alpha(k) \Phi\big((m+1)h-1, k+1\big) \varepsilon^{\prime}(k) \bigg\|^2 \Bigg| \mathcal{F}(mh-1) \Bigg\} \\\nonumber
&
+ \mathbb{E} \Bigg\{ \bigg\| \sum_{k=mh}^{(m+1)h-1} \alpha(k) \Phi\big((m+1)h-1, k+1\big) \varepsilon^{\prime\prime}(k) \bigg\|^2 \Bigg | \mathcal{F}(mh-1) \Bigg\}\\\nonumber
&
+ 2 \mathbb{E} \Bigg\{ \left( \Phi\big((m+1)h-1,mh\big) e(mh) \right)^{\sf T} \\\nonumber
&\quad \left( \sum_{k=mh}^{(m+1)h-1} \alpha(k) \Phi\big((m+1)h-1, k+1\big) \varepsilon^{\prime}(k) \right) \Bigg| \mathcal{F}(mh-1) \Bigg\} \\\nonumber
&
+ 2 \mathbb{E} \Bigg\{\left( \Phi\big((m+1)h-1,mh\big) e(mh) \right)^{\sf T} \\\nonumber
&\quad \left( \sum_{k=mh}^{(m+1)h-1} \alpha(k) \Phi\big((m+1)h-1, k+1\big) \varepsilon^{\prime\prime}(k) \right) \Bigg | \mathcal{F}(mh-1) \Bigg\}\\\nonumber
&
+ 2 \mathbb{E} \Bigg\{\left( \sum_{k=mh}^{(m+1)h-1} \alpha(k) \Phi\big((m+1)h-1, k+1\big) \varepsilon^{\prime}(k) \right)^{\sf T} \\\nonumber
&\quad \left( \sum_{k=mh}^{(m+1)h-1} \alpha(k) \Phi\big((m+1)h-1, k+1\big) \varepsilon^{\prime\prime}(k) \right) \Bigg | \mathcal{F}(mh-1) \Bigg\}\\\nonumber
&:= (A) + (B) + (C) + (D) + (E) + (F) .
\end{align}
For $(A)$, we have
\begin{align}\nonumber
(A) &= \mathbb{E} \Big\{ \left( \Phi\big((m+1)h-1,mh\big) e(mh) \right)^{\sf T} \\\nonumber
&\quad \left( \Phi\big((m+1)h-1,mh\big) e(mh) \right) \big| \mathcal{F}(mh-1) \Big\}\\\nonumber
&= 
e^{\sf T}(mh) \mathbb{E} \left\{   \Phi^{\sf T}\big((m+1)h-1,mh\big)   \Phi\big((m+1)h-1,mh\big)  \big| \mathcal{F}(mh-1) \right\} e(mh)\\\label{eq_use_use_eq_elementary1_b}
&\le 
\exp\left( - c_1 \sum_{k=mh}^{(m+1)h-1} \alpha(k) \right) \|e(mh)\|^2 \le \|e(mh)\|^2,
\end{align}
from \eqref{use_eq_elementary1} and for large enough $m$.

Similar to \eqref{eq_ms_converg_II_3}, we have for $(B)$ that 
\begin{align}\label{eq_linear_recursion_asconvergence_partB_1}
(B) &\le  
c_2 c_{\varepsilon} \sum_{k=mh}^{(m+1)h-1} \alpha^{2(1-\delta)}(k)  \exp\left(- c_1 \sum_{i=k+1}^{(m+1)h-1} \alpha(i)\right)\\\label{eq_linear_recursion_asconvergence_partB_2}
&\le c_2 c_{\varepsilon} \sum_{k=mh}^{(m+1)h-1} \alpha^{2(1-\delta)}(k),
\end{align}
where $\delta$ is given in Assumption~\ref{asmp_linear_recursion}~(iii.a).

Note that $(C)$ can be bounded by
\begin{align}\nonumber
&\sum_{k=mh}^{(m+1)h-1} \sum_{l=mh}^{(m+1)h-1} \alpha(k) \alpha(l) \mathbb{E} \Big\{ \big| \varepsilon^{\prime\prime}(k)^{\sf T} \Phi^{\sf T}\big((m+1)h-1, k+1\big) \\\nonumber
&\quad \Phi\big((m+1)h-1, l+1\big) \varepsilon^{\prime\prime}(l) \big| \Big | \mathcal{F}(mh-1) \Big\}\\\nonumber
&\le
\sum_{k=mh}^{(m+1)h-1} \sum_{l=mh}^{(m+1)h-1} \alpha(k) \alpha(l) \mathbb{E} \Big\{\big\| \Phi\big((m+1)h-1, k+1\big)\varepsilon^{\prime\prime}(k) \big\|  \\\nonumber
&\quad \big\|\Phi\big((m+1)h-1, l+1\big) \varepsilon^{\prime\prime}(l) \big\| \Big | \mathcal{F}(mh-1) \Big\}\\\label{eq_thm_asconverg_conditional_holder}
&\le 
\sum_{k=mh}^{(m+1)h-1} \sum_{l=mh}^{(m+1)h-1} \alpha(k) \alpha(l) \\\nonumber
&\quad \Big(\mathbb{E} \Big\{ \big\| \Phi\big((m+1)h-1, k+1\big) \varepsilon^{\prime\prime}(k) \big\|^2 \Big| \mathcal{F}(mh-1) \Big\} \Big)^{\frac12} \\\nonumber
&\quad \Big(\mathbb{E} \Big\{ \big\|\Phi\big((m+1)h-1, l+1\big) \varepsilon^{\prime\prime}(l) \big\|^2 \Big| \mathcal{F}(mh-1) \Big\} \Big)^{\frac12}\\\label{eq_thm_asconverg_lem4_3}
&=
\left(c_2 \sum_{k=mh}^{(m+1)h-1} \alpha(k) g_2(k) \exp\left(- \frac12 c_1 \sum_{i=k+1}^{(m+1)h-1} \alpha(i)\right) \right)^2\\\label{eq_linear_recursion_as_converg_C_1}
&\le
c_2^2 \left(\sum_{k=mh}^{(m+1)h-1} \alpha(k) g_2(k) \right)^2 \\\nonumber
&\le 
c_2^2 h \sum_{k=mh}^{(m+1)h-1} \alpha^2(k) g_2^2(k), 
\end{align}
where \eqref{eq_thm_asconverg_conditional_holder} follows from the conditional H{\"o}lder inequality, and \eqref{eq_thm_asconverg_lem4_3} is implied by Lemma \ref{lem_linear_recursion} and Assumption~\ref{asmp_linear_recursion}~(iii.b).

From the assumptions of the theorem, $\varepsilon^{\prime}(t)$ can be written as $u(t)w(t)$, and $\{w(t), t\ge mh\}$ is conditionally independent of $\{Q(t), \Delta(t), u(t), t \ge mh\}$ given $\mathcal{F}(mh-1) = \sigma(Q(s), \Delta(s), u(s), w(s), 0\le s \le mh-1)$ from Lemma A.1 in \cite{li2018distributed}. Thus, it follows that
\begin{align*}
(D) &= 2 e(mh) \sum_{k=mh}^{(m+1)h-1} \alpha(k) \mathbb{E} \big\{ \Phi^{\sf T}\big((m+1)h-1,mh\big) \\\nonumber
&\quad \Phi\big((m+1)h-1, k+1\big) \varepsilon^{\prime}(k) \big | \mathcal{F}(mh-1) \big\}\\
&=
2 e(mh) \sum_{k=mh}^{(m+1)h-1} \alpha(k) \mathbb{E} \big\{ \Phi^{\sf T}\big((m+1)h-1,mh\big) \\\nonumber
&\quad \Phi\big((m+1)h-1, k+1\big) u(k) w(k) \big | \mathcal{F}(mh-1) \big\}\\
&=0,
\end{align*}
from the conditional independence and that $\{w(t), \mathcal{F}(t)\}$ is a martingale difference.

By conditional H{\"o}lder inequalities, for $(E)$ we have that
\begin{align}\nonumber
(E) &= 2 \mathbb{E} \Bigg\{\left( \Phi\big((m+1)h-1,mh\big) e(mh) \right)^{\sf T} \\\nonumber
&\quad \left( \sum_{k=mh}^{(m+1)h-1} \alpha(k) \Phi\big((m+1)h-1, k+1\big) \varepsilon^{\prime\prime}(k) \right) \Bigg | \mathcal{F}(mh-1) \Bigg\}\\\nonumber
&\le
2 \Big( \mathbb{E} \big\{ \big\|\Phi\big((m+1)h,mh \big) e(mh)\|^2 \big | \mathcal{F}(mh-1) \big\} \Big)^\frac12 \\\nonumber
&\quad \Bigg( \mathbb{E} \Bigg\{ \Bigg\| \sum_{k=mh}^{(m+1)h-1} \alpha(k) \Phi\big((m+1)h-1, k+1\big) \varepsilon^{\prime\prime}(k) \Bigg\|^2 \Bigg | \mathcal{F}(mh-1) \Bigg\} \Bigg)^\frac12\\\nonumber
&\le 
2 c_2 \|e(mh)\| \sum_{k=mh}^{(m+1)h-1} \alpha(k) g_2(k),
\end{align}
where the last inequality is from \eqref{eq_use_use_eq_elementary1_b} and \eqref{eq_linear_recursion_as_converg_C_1}.

Similarly, it can be obtained from \eqref{eq_linear_recursion_asconvergence_partB_2} and \eqref{eq_linear_recursion_as_converg_C_1} that for large enough $m$
\begin{align*}
(F) &\le 2 c_2 \sqrt{c_{\varepsilon} \sum_{k=mh}^{(m+1)h-1} \alpha^{2(1-\delta)}(k)} \sum_{k=mh}^{(m+1)h-1} \alpha(k) g_2(k)\\
&\le 2 c_2 \sqrt{c_{\varepsilon}} \sum_{k=mh}^{(m+1)h-1} \alpha(k) g_2(k).
\end{align*}

To sum up, for large enough $m$
\begin{align}\nonumber
& \mathbb{E} \{\|e((m+1)h)\|^2 | \mathcal{F}(mh-1) \}\\\nonumber
&\le
\|e(mh)\|^2 + c_4 \Bigg(\sum_{k=mh}^{(m+1)h-1} \alpha^{2(1-\delta)}(k) + \sum_{k=mh}^{(m+1)h-1} \alpha^{2}(k) g^2_2(k) \\\nonumber
&\quad + (\|e(mh)\| + c_5) \sum_{k=mh}^{(m+1)h-1} \alpha(k) g_2(k) \Bigg),
\end{align}
where $c_4$ and $c_5$ are positive constants.

We have proved that $\mathbb{E}\{\|e(t)\|^2\} \to 0$, yielding that $\mathbb{E}\{\|e(mh)\|\} \to 0$ as $m \to \infty$. Thus, $\mathbb{E} \left\{\|e(mh)\| \right\}$ is bounded, and from $g_2(t) \to 0$ it follows that
\begin{align*}
&\mathbb{E} \Bigg\{\sum_{m=0}^{\infty} \Bigg( \sum_{k=mh}^{(m+1)h-1} \alpha^{2(1-\delta)}(k) + \sum_{k=mh}^{(m+1)h-1} \alpha^{2}(k) g^2_2(k) \\\nonumber
&\quad + (\|e(mh)\| + c_5) \sum_{k=mh}^{(m+1)h-1} \alpha(k) g_2(k) \Bigg) \Bigg\}\\
&= 
\sum_{m=0}^{\infty} \Bigg(  \sum_{k=mh}^{(m+1)h-1} \alpha^{2(1-\delta)}(k) + \sum_{k=mh}^{(m+1)h-1} \alpha^{2}(k) g^2_2(k) \\\nonumber
&\quad + (\mathbb{E} \left\{\|e(mh)\| \right\} + c_5) \sum_{k=mh}^{(m+1)h-1} \alpha(k) g_2(k) \Bigg)\\
&\le
\sum_{m=0}^{\infty} \left(c_6 \sum_{k=mh}^{(m+1)h-1} \alpha^{2(1-\delta)}(k) + c_7 \sum_{k=mh}^{(m+1)h-1} \alpha(k) g_2(k) \right)\\
&=
c_6 \sum_{k=0}^{\infty} \alpha^{2(1-\delta)}(k) + c_7 \sum_{k=0}^{\infty} \alpha(k) g_2(k),
\end{align*}
where $c_6$ and $c_7$ are positive constants. From the assumptions we know the above series converges a.s, so Lemma 1.2.2 in \cite{chen2002stochastic} ensures that $e(mh)$ converges a.s as $m \to \infty$, implying that $e(mh) \to 0$ a.s.

Noting that for $\forall\varepsilon > 0$,
\begin{align*}
\begin{split}
&\sum_{k=0}^{\infty} \mathbb{P}\{\alpha(k) \|\varepsilon^{\prime}(k)\| \ge \varepsilon\}\\
&\le \frac{1}{\varepsilon^2} \sum_{k=0}^{\infty} \alpha^{2(1-\delta)}(k) \mathbb{E}\{\|\alpha^{\delta}(k) \varepsilon^{\prime}(k)\|^2\}\\
&\le \frac{c_{\varepsilon}}{\varepsilon^2} \sum_{k=0}^{\infty} \alpha^{2(1-\delta)}(k) < \infty,
\end{split}
\end{align*}
from Chebyshev inequality. Hence by Borel-Cantelli lemma, $\alpha(k) \|\varepsilon^{\prime}(k)\| \to 0$ a.s. 
Note that for $m \ge 0$ and $0 < s < h$, 
\begin{align*}
\left\|e(mh+s) \right\| &\le \left\|\Phi(mh+s-1,mh) e(mh)\right\| + \left\|\sum_{k=mh}^{mh+s-1} \alpha(k) \Phi(mh+s, k+1) \varepsilon^{\prime}(k)\right\| \\
& + \left\|\sum_{k=mh}^{mh+s-1} \alpha(k) \Phi(mh+s, k+1) \varepsilon^{\prime\prime}(k)\right\|,
\end{align*}
and the right side converges to zero a.s. for fixed $s$ as $m \to \infty$, from $e(mh) \to 0$, $\alpha(k) \|\varepsilon^{\prime}(k)\| \to 0$, Assumption~\ref{asmp_linear_recursion} (ii.a) and (iii.b). Therefore, we prove that $e(t) \to 0$ a.s.
	
	\subsection{Proof of Theorem  \ref{thm_mean_square_convergence}}\label{pf_thm_mean_square_convergence}
	
	From \eqref{eq_estimation_compact} and $\mathcal{L}\mathbf{1}_N = \mathbf{0}_N$, the estimation error $e_0(t)=X(t)-\textbf{1}_N\otimes \theta$ evolves as below,  
	\begin{equation*}
	\begin{aligned}
	e_0(t+1) 
	&= 
	e_0(t) - \bar{\alpha}(t) (\mathcal{L}\otimes I_M + \bar D_H(t) \bar D_H^{\sf T}(t)) e_0(t) \\
	&\quad + \bar{\alpha}(t) \left(\bar D_H(t) V(t) + (\mathcal{A}\otimes I_M)(X(\tau_{k})-X(t)) \right)\\
	&= 
	e_0(t) - \alpha(t) \bigg(\frac{\bar{\alpha}(t)}{\alpha(t)}\bigg) (\mathcal{L}\otimes I_M + \bar D_H(t) \bar D_H^{\sf T}(t)) e_0(t) \\
	&\quad + \alpha(t) \bigg(\frac{\bar{\alpha}(t)}{\alpha(t)}\bigg) \left(\bar D_H(t) V(t) + (\mathcal{A}\otimes I_M)(X(\tau_{k})-X(t)) \right)\\
	&=
	e_0(t) - \alpha(t) (I_{MN} + o(1)) (\mathcal{L}\otimes I_M + \bar D_H(t) \bar D_H^{\sf T}(t)) e_0(t) \\
	&\quad + \alpha(t) (I_{MN} + o(1)) \left(\bar D_H(t) V(t) + (\mathcal{A}\otimes I_M)(X(\tau_{k})-X(t)) \right),
	\end{aligned}
	\end{equation*}
	where $o(1)$ is a deterministic infinitesimal.
	
	Let $e(t) = e_0(t)/\alpha^{\delta}(t)$, $\delta \in [0,1/2)$, so
	\begin{align}\nonumber
	&e(t+1) \\\nonumber
	&=
	\left( \frac{\alpha(t)}{\alpha(t+1)} \right)^{\delta} \frac{e_0(t+1)}{\alpha^{\delta}(t)} \\\nonumber
	&=
	\left( \frac{\alpha(t)}{\alpha(t+1)} \right)^{\delta} \frac{1}{\alpha^{\delta}(t)} \Big( \big( I_{MN} \\\nonumber
	&\quad - \alpha(t) (I_{MN} + o(1)) (\mathcal{L}\otimes I_M + \bar D_H(t) \bar D_H^{\sf T}(t)) \big) e_0(t) \\\nonumber
	&\quad+ \alpha(t) (I_{MN} + o(1)) \big(\bar D_H(t) V(t) + (\mathcal{A}\otimes I_M)(X(\tau_{k})-X(t)) \big) \Big) \\\nonumber
	&=
	\left( \frac{\alpha(t)}{\alpha(t+1)} \right)^{\delta} \bigg( \big( I_{MN} \\\nonumber
	&\quad - \alpha(t) (I_{MN} + o(1)) (\mathcal{L}\otimes I_M + \bar D_H(t) \bar D_H^{\sf T}(t)) \big) e(t) \\\nonumber
	&\quad  + \alpha(t) (I_{MN} + o(1)) \bigg(\frac{\bar D_H(t) V(t)}{\alpha^{\delta}(t)}
	+ (\mathcal{A}\otimes I_M)\frac{X(\tau_{k})-X(t)}{\alpha^{\delta}(t)} \bigg) \bigg)\\\label{eq_thm_mean_square_convergence_1}
	&=
	\left(1 + \delta \frac{\alpha(t) - \alpha(t+1)}{\alpha(t+1)} + O\left(\left(\frac{\alpha(t) - \alpha(t+1)}{\alpha(t+1)}\right)^2\right) \right) \\\nonumber
	&\quad~ \bigg( \big( I_{MN} - \alpha(t) (I_{MN} + o(1)) (\mathcal{L}\otimes I_M + \bar D_H(t) \bar D_H^{\sf T}(t)) \big) e(t) \\\nonumber
	&\quad + \alpha(t) (I_{MN} + o(1)) \left(\frac{\bar D_H(t) V(t)}{\alpha^{\delta}(t)} + (\mathcal{A}\otimes I_M)\frac{X(\tau_{k})-X(t)}{\alpha^{\delta}(t)} \right) \bigg)\\\label{eq_thm_mean_square_convergence_2}
	&=
	\left(1 + \delta \frac{\alpha(t) - \alpha(t+1)}{\alpha(t+1)} + O\left(\left(\frac{\alpha(t) - \alpha(t+1)}{\alpha(t+1)}\right)^2\right) \right) \\\nonumber
	&\quad~ \big( I_{MN} - \alpha(t) (I_{MN} + o(1)) (\mathcal{L}\otimes I_M + \bar D_H(t) \bar D_H^{\sf T}(t)) \big) e(t)\\\nonumber
	&\quad +  \big(1 + O(\alpha(t)) \big) \alpha(t) (I_{MN} + o(1)) \\\nonumber
	&\quad \left(\frac{\bar D_H(t) V(t)}{\alpha^{\delta}(t)} + (\mathcal{A}\otimes I_M)\frac{X(\tau_{k})-X(t)}{\alpha^{\delta}(t)} \right)\\\nonumber
	&=
	\Bigg\{I_{MN} + \alpha(t) \\\nonumber
	&\qquad \bigg[ - (I_{MN} + o(1)) (\mathcal{L}\otimes I_M + \bar D_H(t) \bar D_H^{\sf T}(t))  + \delta \frac{\alpha(t) - \alpha(t+1)}{\alpha(t)\alpha(t+1)} I_{MN} \\\nonumber
	&\qquad~~  - \delta \frac{\alpha(t) - \alpha(t+1)}{\alpha(t+1)} (I_{MN} + o(1)) (\mathcal{L}\otimes I_M + \bar D_H(t) \bar D_H^{\sf T}(t) ) \\\nonumber
	&\qquad~~ + O\left(\left(\frac{\alpha(t) - \alpha(t+1)}{\alpha(t+1)}\right)^2\right) \bigg( \frac{1}{\alpha(t)} I_{MN} \\\nonumber
	&\qquad~~ - (I_{MN} + o(1)) (\mathcal{L}\otimes I_M + \bar D_H(t) \bar D_H^{\sf T}(t)) \bigg) \bigg] \Bigg\} e(t)\\\nonumber
	&\qquad + \alpha(t) \big(1 + O(\alpha(t))\big) (I_{MN} + o(1)) \\\nonumber
	&\qquad \left(\frac{\bar D_H(t) V(t)}{\alpha^{\delta}(t)} + (\mathcal{A}\otimes I_M)\frac{X(\tau_{k})-X(t)}{\alpha^{\delta}(t)} \right),
	\end{align}
	where \eqref{eq_thm_mean_square_convergence_1} follows from Taylor's expansion and \eqref{eq_lem_linear_recursion_stepsize1}, and \eqref{eq_thm_mean_square_convergence_2} is from \eqref{eq_lem_linear_recursion_stepsize2}.
	
	Now let
	\begin{align*}
	Q(t) 
	&= - (\mathcal{L}\otimes I_M + \bar D_H(t) \bar D_H^{\sf T}(t)) + \alpha_0 \delta I_{MN},\\
	\Delta(t) 
	&= o(1) (\mathcal{L} \otimes I_M + \bar{D}_H(t) \bar{D}^{\sf T}_H(t)) + \delta \left(\frac{\alpha(t) - \alpha(t+1)}{\alpha(t)\alpha(t+1)} - \alpha_0 \right) I_{MN} \\
	&\quad - \delta \frac{\alpha(t) - \alpha(t+1)}{\alpha(t+1)} (I_{MN} + o(1)) \big(\mathcal{L}\otimes I_M + \bar D_H(t) \bar D_H^{\sf T}(t) \big)\\
	&\quad +O\left(\left(\frac{\alpha(t) - \alpha(t+1)}{\alpha(t+1)}\right)^2\right) \bigg( \frac{1}{\alpha(t)} I_{MN} \\\nonumber
	&\quad - (I_{MN} + o(1)) ( \mathcal{L}\otimes I_M + \bar D_H(t) \bar D_H^{\sf T}(t))\bigg),\\
	\varepsilon^{\prime}(t) 
	&= \big(1 + O(\alpha(t))\big) (I_{MN} + o(1)) \frac{\bar D_H(t) V(t)}{\alpha^{\delta}(t)},\\
	\varepsilon^{\prime\prime}(t) 
	&= \big(1 + O(\alpha(t))\big) (I_{MN} + o(1)) (\mathcal{A}\otimes I_M)\frac{X(\tau_{k})-X(t)}{\alpha^{\delta}(t)},
	\end{align*}
	and we can write $e_0(t)/\alpha^{\delta}(t)$ in the form \eqref{eq_linear_recursion}. Note that Assumption~\ref{asmp_consistency2}~(i.a) is identical to Assumption~\ref{asmp_linear_recursion}~(i). Since
	\begin{align*}
	&Q(t) + Q^{\sf T}(t)
	=
	- (\mathcal{L} + \mathcal{L}^{\sf T}) \otimes I_M - 2 \bar D_H(t) \bar D_H^{\sf T}(t) + 2 \alpha_0 \delta I_{MN},
	\end{align*}
	Assumption~\ref{asmp_linear_recursion}~(ii.a) holds under Assumption~\ref{asmp_consistency2}~(ii.c). 
	
	From \eqref{eq_lem_linear_recursion_stepsize2}, it holds that
	\begin{align*}
	\Delta(t) &= o(1) (\mathcal{L} \otimes I_M + \bar{D}_H(t) \bar{D}^{\sf T}_H(t)) + \delta \left(\frac{\alpha(t) - \alpha(t+1)}{\alpha(t)\alpha(t+1)} - \alpha_0 \right) I_{MN} \\
	&\quad - \delta \frac{\alpha(t) - \alpha(t+1)}{\alpha(t+1)} (I_{MN} + o(1)) \big(\mathcal{L}\otimes I_M + \bar D_H(t) \bar D_H^{\sf T}(t) \big)\\
	&\quad  + O\left(\frac{\alpha(t) - \alpha(t+1)}{\alpha(t+1)}\right) I_{MN}\nonumber\\
	&\quad - O\left(\left(\frac{\alpha(t) - \alpha(t+1)}{\alpha(t+1)}\right)^2\right) (I_{MN} + o(1)) \left(\mathcal{L}\otimes I_M + \bar D_H(t) \bar D_H^{\sf T}(t) \right),
	\end{align*}
	so by $\sup_{t\in\mathbb{N}} \|\bar D_H(t)\|^2 \le D$ and Assumption~\ref{asmp_consistency2}~(i.a) we know that $\Delta(t)$ can be bounded by a deterministic function that converges to zero as $t \to \infty$.
	
	Assumption~\ref{asmp_consistency2}~(ii.a) and~(ii.b) ensure that $\{\varepsilon^{\prime}(t), \mathcal{F}(t)\}$ is a martingale difference sequence, and
	\begin{align*}
	&\mathbb{E}\{\|\alpha^{\delta}(t)\varepsilon^{\prime}(t)\|^2|\mathcal{F}(t-1)\} \\
	&=
	\mathbb{E}\big\{ \big\|\big(1 + O(\alpha(t))\big) (I_{MN} + o(1)) \bar D_H(t) V(t) \big\|^2 \big|\mathcal{F}(t-1) \big\} \\
	&\le 
	c_1 \mathbb{E}\{ \| \bar D_H(t) V(t)\|^2 |\mathcal{F}(t-1)\}\\
	&\le 
	c_1 D \mathbb{E}\{ \|V(t)\|^2 |\mathcal{F}(t-1)\} \\
	&\le 
	c_1 D \left( \mathbb{E}\{ \|V(t)\|^{\rho} |\mathcal{F}(t-1)\}\right)^{\frac{2}{\rho}} \le c (c_V)^{\frac{2}{\rho}} D,
	\end{align*}
	where $c_1$ is a positive constant, and the last inequality is obtained from conditional Lyapunov inequality.
	
	Finally, $\|\varepsilon^{\prime\prime}(t)\| \le c_2 f_{\max}(t) / \alpha^{\delta}(t) \to 0$ for some positive constant $c_2$, from Assumption~\ref{asmp_consistency2}~(iii.a). Therefore, Theorem~\ref{thm_linear_recursion} implies the conclusion.

	\subsection{Proof of Theorem \ref{thm_conver_rate}}\label{pf_thm_conver_rate}
	
	Following Section~\ref{pf_thm_mean_square_convergence}, we only need to validate the additional assumptions in Theorem~\ref{thm_linear_recursion}. Note that $\sum_{t=1}^{\infty} \alpha(t)^{2(1-\delta)} < \infty$ is given in Assumption~\ref{asmp_consistency2}~(i.b), and $\sum_{t=1}^{\infty} \alpha(t) g_2(t) < \infty$ holds from Assumption~\ref{asmp_consistency2}~(iii.b), by noticing $g_2(t) := c_2 f_{\max}(t)/\alpha^{\delta}(t)$. Finally, letting $u(t) := \big(1 + O(\alpha(t))\big) (I_{MN} + o(1)) \bar D_H(t)/\alpha^{\delta}(t)$ and $w(t) := V(t)$, we know that (c) holds under Assumption~\ref{asmp_consistency2}~(ii.a) and~(ii.b).

	\subsection{Proof of Theorem \ref{thm_rate}}\label{pf_thm_rate}
	Recall that $\tau_{k}^i:=\tau_{k_i(t)}$ is $k$-th triggering instant of sensor $i\in\mathcal{V}$ in the time interval $[0,t]\cap\mathbb{N}$.
	Denote the set $\Gamma=\{i\in\mathcal{V}| \tau_{k}^i\rightarrow\infty$ \text{as} $t\rightarrow\infty\}$, which is the set of sensors whose number of communications  goes to infinity as time goes to infinity. If $\Gamma=\emptyset$, then there are positive integers $n_0$ and $N_0$, such that for $t\geq n_0$, 
	$\max_{i\in\mathcal{V}}\sup_{t}\tau_{k}^i\leq N_0$ surely.
	From the definition of communication rate \eqref{eq_rate}, for $t\geq n_0$ and $\gamma\in[0,1/2)$, it holds that 
	$$\lim_{t\to \infty} \lambda_c(t)t^{\gamma} \leq \lim_{t\to \infty} \frac{N_0}{t^{1-\gamma}}=0.$$
	In the case, the conclusion holds.
	
	Next, we consider the case that $\Gamma\neq \emptyset$.   
	According to (\ref{eq_estimator02}), for  any sensor $i\in\Gamma$ and time $t\geq\tau_{k}^i$, we have 
	\begin{align}\label{eq_estimator01}
	&x_{i}(t+1)-x_{i}(\tau_{k}^i)\nonumber\\
	=&x_{i}(t)-x_{i}(\tau_{k}^i)+\alpha_i(t)\sum_{j\in\mathcal{N}_{i}}a_{i,j}(x_{j}(t)-x_{i}(t))\\
	&+\alpha_i(t)  H_{i}^{\sf T}(t)( y_{i}(t)-H_{i}(t)x_{i}(t))+\alpha_i(t)\sum_{j\in\mathcal{N}_{i}}a_{i,j}(x_{j}(\tau_{k}^j)-x_{j}(t))\nonumber.
	\end{align}
	From Assumption \ref{asmp_consistency2} (i.a), there is a constant $c_0>0$ such that $\alpha_i(t)\leq c_0\alpha(t).$ 
	Taking norm operator on both sides of (\ref{eq_estimator01}) yields
	\begin{align}\label{eq_estimator012}
	&\norm{x_{i}(t+1)-x_{i}(\tau_{k}^i)}\nonumber\\
	\leq&\norm{x_{i}(t)-x_{i}(\tau_{k}^i)}+ c_0\alpha(t)\norm{\sum_{j\in\mathcal{N}_{i}}a_{i,j}(x_{j}(t)-x_{i}(t))}\nonumber\\
	&+ c_0\alpha(t)\norm{H_{i}^{\sf T}(t)( y_{i}(t)-H_{i}(t)x_{i}(t))}+ c_0\alpha(t)\norm{\sum_{j\in\mathcal{N}_{i}}a_{i,j}(x_{j}(\tau_{k}^j)-x_{j}(t))}\nonumber\\
	:=&\norm{x_{i}(t)-x_{i}(\tau_{k}^i)}+(I)+(II)+(III).
	\end{align}
	
	Consider $(I)$ in \eqref{eq_estimator012}.
	Denote $d_0=\max_{i\in \mathcal{V}}\sum_{j\in\mathcal{N}_i} a_{i,j}$, then we have 
	\begin{align}
	c_0\alpha(t)\norm{a_{i,j}\sum_{j\in\mathcal{N}_{i}}(x_{j}(t)-x_{i}(t))}&= c_0\alpha(t)\norm{\sum_{j\in\mathcal{N}_{i}}a_{i,j}(x_{j}(t)-\theta+\theta-x_{i}(t))}\nonumber\\
	&\leq 2d_0 c_0\alpha(t)\max_{j\in\mathcal{V}}\norm{x_{j}(t)-\theta}\leq \bar c_1 \alpha(t)^{1+\delta}\label{ineq1}
	\end{align}
	where $\bar c_1$ is a positive scalar, which could be different in different sample trajectories, and 
	the last inequality is obtained from Theorem \ref{thm_conver_rate}.

	Consider $(II)$ in \eqref{eq_estimator012}. 
	For $\forall\varepsilon_0 > 0$,   by (ii.b) of Assumption  \ref{asmp_consistency2} and Markov inequality, we obtain
	\begin{align}\label{pro_ineq2}
	\begin{split}
	\sum_{t=0}^{\infty} \mathbb{P}\{\alpha(t)^{\frac{2(1-\delta)}{\rho}} \|V(t)\| \ge \varepsilon_0\}&\le \frac{1}{\varepsilon_0^{\rho}} \sum_{t=0}^{\infty} \alpha(t)^{2(1-\delta)} \mathbb{E}\{\|V(t)\|^\rho\}\\
	&\le \frac{\bar c_V}{\varepsilon_0^\rho} \sum_{t=0}^{\infty} \alpha(t)^{2(1-\delta)} < \infty.
	\end{split}
	\end{align}
	Hence by Borel-Cantelli lemma, we have
	\begin{align}\label{eq_g}
	\lim\limits_{t\rightarrow\infty}\alpha(t)^{\frac{2(1-\delta)}{\rho}}\norm{v_i(t)}=0, \text{ a.s.}
	\end{align}
	Under (ii.c) of Assumption \ref{asmp_consistency2}, $\sup_{t\geq 0}\norm{H_i^{\sf T}(t)}<\infty,$ then from  \eqref{eq_g} and  Theorem \ref{thm_conver_rate}, there are positive scalar $\bar c_2,\bar c_3$,  which could be different in different sample trajectories, such that
	\begin{align}
	&c_0\alpha(t)\norm{H_{i}^{\sf T}(t)( y_{i}(t)-H_{i}(t)x_{i}(t))}\nonumber\\
	=& c_0\alpha(t)\norm{H_{i}^{\sf T}(t)H_{i}(t)( \theta-x_{i}(t))+H_{i}^{\sf T}(t)v_i(t)}\nonumber\\
	\leq & c_0\alpha(t)\max_{i\in\mathcal{V}}\sup_{t\geq 0}\norm{H_{i}^{\sf T}(t)H_{i}(t)}\max_{i\in\mathcal{V}}\norm{x_{i}(t)-\theta}+\bar c_2\alpha(t)^{1-\frac{2(1-\delta)}{\rho}}\nonumber\\
	\leq &\bar c_3\alpha(t)^{1+\delta}+\bar c_2\alpha(t)^{1-\frac{2(1-\delta)}{\rho}}.\label{ineq2}
	\end{align}
	
	Consider $(III)$ in \eqref{eq_estimator012}. Since $\norm{x_{j}(\tau_{k}^j)-x_{j}(t)}\leq f_j(t)$, 
	\begin{align}
	c_0\alpha(t)\norm{\sum_{j\in\mathcal{N}_{i}}a_{i,j}(x_{j}(\tau_{k}^j)-x_{j}(t))}\leq  c_0d_0\alpha(t)f_{\max}(t).\label{ineq3}
	\end{align}
	
	From  inequalities \eqref{eq_estimator012}, \eqref{ineq1}, \eqref{ineq2} and \eqref{ineq3}, it follows that
	\begin{align*}
	\norm{x_{i}(t+1)-x_{i}(\tau_{k}^i)}\leq &\norm{x_{i}(t)-x_{i}(\tau_{k}^i)}+(\bar c_1 + \bar c_3)\alpha(t)^{1+\delta}\\
	&+\bar c_2\alpha(t)^{1-\frac{2(1-\delta)}{\rho}}+ c_0d_0\alpha(t)f_{\max}(t).
	\end{align*}
	Under Assumption \ref{ass_triggering}, there is a monotonically non-increasing sequence $	\beta(t)=O(\alpha(t)^{1-2(1-\delta)/\rho}).$
	Thus, there is a scalar $\bar c_4>0$, which could be different in different sample trajectories, such that 
	\begin{align*}
	\norm{x_{i}(t+1)-x_{i}(\tau_{k}^i)}\leq \norm{x_{i}(t)-x_{i}(\tau_{k}^i)}+ \bar c_4 \beta(t).
	\end{align*} 
	Denote
	$L_{k}^i:=\tau_{k+1}^i-\tau_{k}^i$ the interval length between $(k+1)$-th triggering time and $k$-th triggering time,  then  we have
	%
	\begin{align*}
	\norm{x_i(\tau_{k+1}^i)-x_i(\tau_{k}^i)}
	=&\norm{x_i(\tau_{k}^i+L_{k}^i)-x_i(\tau_{k}^i)}\\
	\leq&\norm{x_i(\tau_{k}^i+L_{k}^i-1)-x_i(\tau_{k}^i)}+\bar c_4 \beta(\tau_{k}^i+L_{k}^i-1)\\
	&\qquad\vdots\\
	\leq&\bar c_4\sum_{s=\tau_{k}^i}^{\tau_{k}^i+L_{k}^i-1} \beta(s)
	\leq\bar c_4L_{k}^i\beta(\tau_{k}^i).
	\end{align*}
	where the last inequality is obtained from the monotonicity of $\beta(t)$.

	%

	A necessary condition to guarantee that the event  is triggered for sensor $i$ is 
	\begin{align} \label{equiv_tri}
	\bar c_4L_{k}^i\beta(\tau_{k}^i)>f_i(\tau_{k+1}^i)
	\Longleftrightarrow&L_{k}^i>\frac{f_i(\tau_{k}^i+L_{k}^i)}{\bar c_4\beta(\tau_{k}^i)}. 
	\end{align}
	
	{
		Then we make the claim:
		\begin{align}\label{calim}
		\liminf\limits_{k\rightarrow\infty}\frac{L_{k}^i}{(\tau_{k}^i)^{\mu}}>0,
		\end{align}
		where $\mu\in[1/2,1)$ is introduced in Assumption~\ref{ass_triggering}.
		The proof of claim \eqref{calim} is given by contradiction. Suppose    claim \eqref{calim} does not hold, i.e., $\liminf\limits_{k\rightarrow\infty}\frac{L_{k}^i}{(\tau_{k}^i)^{\mu}}=0.$ Then $\{k\}_{k=0}^{\infty}$ has a subsequence $\{k_{j}\}_{j=0}^{\infty}$ such that
		\begin{align}\label{contradiction}
		\lim\limits_{j\rightarrow\infty}\frac{L_{k_j}^i}{(\tau_{k_j}^i)^{\mu}}=0,
		\end{align}
		which means there is a finite integer $J>0$, such that $L_{k_j}^i\leq (\tau_{k_j}^i)^{\mu}$ for any $j\geq J$. It follows from \eqref{equiv_tri} and   the monotonicity of $\bar f(t)$ that 
		\begin{align*}
		L_{k_j}^i&>\frac{\bar f(\tau_{k_j}^i+L_{k_j}^i)}{\bar c_4\beta(\tau_{k_j}^i)}=\frac{\bar f(\tau_{k_j}^i+L_{k_j}^i)}{\bar c_4\bar f(\tau_{k_j}^i)}\frac{\bar f(\tau_{k_j}^i)}{\beta(\tau_{k_j}^i)}\geq \frac{\bar f(\tau_{k_j}^i+(\tau_{k_j}^i)^{\mu})}{\bar c_4\bar f(\tau_{k_j}^i)}\frac{\bar f(\tau_{k_j}^i)}{\beta(\tau_{k_j}^i)}\geq  \frac{\bar c_5}{\bar c_4}\frac{\bar f(\tau_{k_j}^i)}{\beta(\tau_{k_j}^i)},
		\end{align*}
		where $\bar c_5>0$ exists due to Assumption \ref{ass_triggering} (ii). From  Assumption \ref{ass_triggering} (iv),  there is $\bar c_6>0$, such that $\frac{\bar f(t)}{\beta(t)}>\bar c_6t^{\mu}$ for any $t\in\mathbb{N}$. Then $L_{k_j}^i>\frac{\bar c_5\bar c_6}{\bar c_4}(\tau_{k_j}^i)^{\mu}$ for $j\geq J$, which contradicts  \eqref{contradiction}. Thus,  claim \eqref{calim} holds.

		According to \eqref{calim}, there is $\bar M>0$, such that $L_{k}^i> \bar M (\tau_{k}^i)^{\mu}=:g(\tau_{k}^i)$ for any $k\in\mathbb{N}$. It follows that 
		\begin{align*} 
		g(\tau_{k+1}^i)-g(\tau_{k}^i)=\bar M (\tau_{k}^i+L_{k}^i)^{\mu}-\bar M (\tau_{k}^i)^{\mu}
		\geq \bar M (\tau_{k}^i+\bar M(\tau_{k}^i)^{\mu})^{\mu}-\bar M (\tau_{k}^i)^{\mu}.
		\end{align*}
		According to the Newton's generalized binomial theorem and $\mu\in[1/2,1)$, there are $T_1>0$ and $\tilde M>0$, such that for $\tau_{k}^i\geq T_1$, we have 
		$\bar M (\tau_{k}^i+\bar M(\tau_{k}^i)^{\mu})^{\mu}-\bar M (\tau_{k}^i)^{\mu}\geq \tilde M (\tau_{k}^i)^{2\mu-1}$.
		%
		%
		%
		Therefore, for $\tau_{k}^i\geq T_1$, it holds that
		\begin{align}\label{g_seq}
		L_{k}^i>g(\tau_{k}^i), \quad g(\tau_{k+1}^i)-g(\tau_{k}^i)\geq \tilde M (\tau_{k}^i)^{2\mu-1}.
		\end{align}
	}
	
	Next, we prove the decay speed of the communication rate.  In the following, we consider the communication frequency of the sensor  which is assumed to have the most triggering time instants during $[0,t]$. 
	Without losing generality, we denote this sensor as sensor $i$. Then it holds that $i\in\Gamma$, which means $\tau_{k}^i\rightarrow\infty$ as $t\rightarrow\infty$.
	There is an integer $s>0$, such that $\tau_{s}^i\geq T_1$. Split the interval $[0,t]\cap\mathbb{N}$ into two sub-intervals, namely $[0,\tau_{s}^i]\cap\mathbb{N}$ and $(\tau_{s}^i,t]\cap\mathbb{N}$. 
	Denote  $\bar s$   the triggering times of sensor $i$ in $(\tau_{s}^i,t]\cap\mathbb{N}$. 
	According to the definition of communication rate in Definition \ref{def_com_rate}, we have 
	\begin{align*}
	\lambda_c(t)=\frac{\sum_{j\in\mathcal{V}}K_j(t)|\mathcal{N}_{j}^c| }{t \sum_{j\in\mathcal{V}}|\mathcal{N}_{j}^c|}
	\leq \frac{K_i(t)}{t}
	=\frac{s+\bar s}{t}.
	\end{align*}
	
	It follows that   for any $\gamma\in[0,\frac{2\mu}{2\mu+1})$,  
	\begin{align}\label{pf_0}
	\lambda_c(t) t^{\gamma}\leq \frac{s}{t^{1-\gamma}}+\frac{\bar s}{t^{1-\gamma}}.
	\end{align}
	We assume $\bar s\rightarrow \infty$ as $t\rightarrow \infty$. Otherwise, the conclusion of this theorem holds trivially.
	Given the above finite $s$, it is straightforward to see that
	\begin{align}\label{pf_1}
	\frac{s}{t^{1-\gamma}}\rightarrow 0 \quad  t\rightarrow\infty.
	\end{align}
	
	{
		Recall that $\tau_{k}^i$ is $k$-th triggering instant of sensor $i\in\mathcal{V}$ in the time interval $[0,t]\cap\mathbb{N}$, thus we have	$\tau_{k}^i\geq k$. Then 
		for $\tau_{l}^i\geq \tau_{s}^i\geq T_1$,   it follows from \eqref{g_seq}  that
		\begin{align}
		g(\tau_{l}^i)&=g(\tau_{s}^i)+\sum_{k=s}^{l-1}(g(\tau_{k+1}^i)-g(\tau_{k}^i))\nonumber\\
		&\geq g(\tau_{s}^i)+\tilde M\sum_{k=s}^{l-1} (\tau_{k}^i)^{2\mu-1}\nonumber\\
		&=g(\tau_{s}^i)+\tilde M\sum_{k=1}^{l-1} k^{2\mu-1}-\tilde M\sum_{k=1}^{s-1} k^{2\mu-1}\nonumber\\
		&\geq g(\tau_{s}^i)+\tilde M_2l^{2\mu}-\tilde M\sum_{k=1}^{s-1} k^{2\mu-1},\nonumber\\
		&\geq \tilde M_3l^{2\mu}\label{ineq_g}
		\end{align}
		where $\tilde M_2,\tilde M_3>0$, and  the last inequality  is obtained by using the sum formula in  \cite[page 1]{gradshteyn2014table}.

		Next, we consider $\bar s/t^{1-\gamma}$. According to \eqref{g_seq}, \eqref{ineq_g}, and $\tau_{l}^i\geq\tau_{s}^i\geq T_1$ for $l\geq s$,  it follows that
		\begin{align*}
		t\geq t-\tau_s^i\geq \sum_{l=s}^{s+\bar s-1}L_{l}^i\geq \sum_{l=s}^{s+\bar s-1} g(\tau_{l}^i)\geq  \tilde M_3 \sum_{l=s}^{s+\bar s-1}l^{2\mu}\geq \tilde M_4(s+\bar s-1)^{2\mu+1},
		\end{align*}
		where $\tilde M_4>0$  and  the last inequality  is obtained by using the sum formula in  \cite[page 1]{gradshteyn2014table}.

		Due to $\gamma\in[0,\frac{2\mu}{2\mu+1})$, we use the  L'Hospital's Rule to obtain
		\begin{align}\label{pf_2}
		\frac{\bar s}{t^{1-\gamma}}\leq  \frac{\bar s}{\left(\tilde M_4(s+\bar s-1)^{2\mu+1}\right)^{1-\gamma}}\overset{\bar s\rightarrow\infty}{\longrightarrow} 0.
		\end{align}
		
	}
	
	From \eqref{pf_0}, \eqref{pf_1}, and \eqref{pf_2}, it follows that
	$\lambda_c(t) t^{\gamma}\rightarrow 0, t\rightarrow\infty,$
	which leads to the conclusion of this theorem.
	
	\section{Numerical Simulations}\label{sec:simulation}
	In this section, we provide two examples to illustrate the effectiveness of Algorithm \ref{alg:A}  and the developed theoretical results.
	
	\subsection{Example 1}
	In this example, we consider the sensor  network in FIG. \ref{fig:diag} with $N=7$ sensors.
	Suppose the  parameter vector to be estimated is  $\theta=[\theta_1,\theta_2]^{\sf T},$ where $\theta_1=-1$ and $\theta_2=2.$ The sensor measurement matrices and the initial  estimates  are in the following
	\begin{align*}
	&H_1=[1,0]^{\sf T},\quad H_2=[0,1]^{\sf T},\quad H_7=H_5=H_3=H_1,\quad H_6=H_4=H_2,\\
	&x_{1}(0)=[0,-100]^{\sf T},\quad x_{7}(0)=x_{5}(0)=x_{3}(0)=x_{1}(0),\\
	&x_{2}(0)=[100,0]^{\sf T},\quad  x_{6}(0)=x_{4}(0)=x_{2}(0).
	\end{align*}
	Suppose the time interval is from $t=0$ to   $t=1000$.  
	The   noise of each sensor follows an Gaussian process with   mean zero and standard deviation 0.1.
	The noise processes are  independent in time and space. 
	
	\begin{figure*}[t]
		\centering
		\subfigure[\label{fig:times2}Element-wise asymptotic convergence.]{\includegraphics[width=0.49\linewidth]{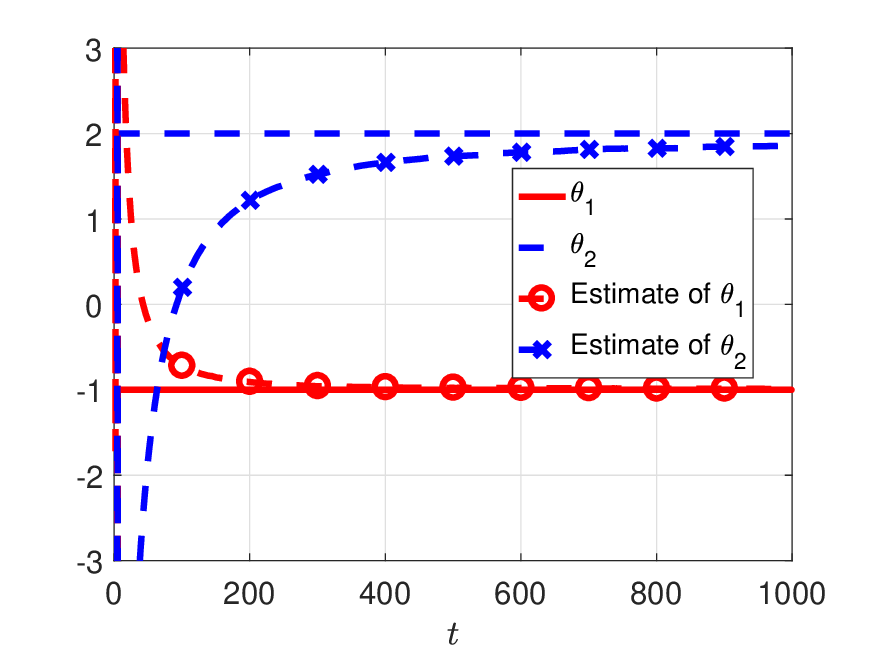}}
		\subfigure[\label{fig:times} Communication triggering instants of each sensor.]{\includegraphics[width=0.49\linewidth]{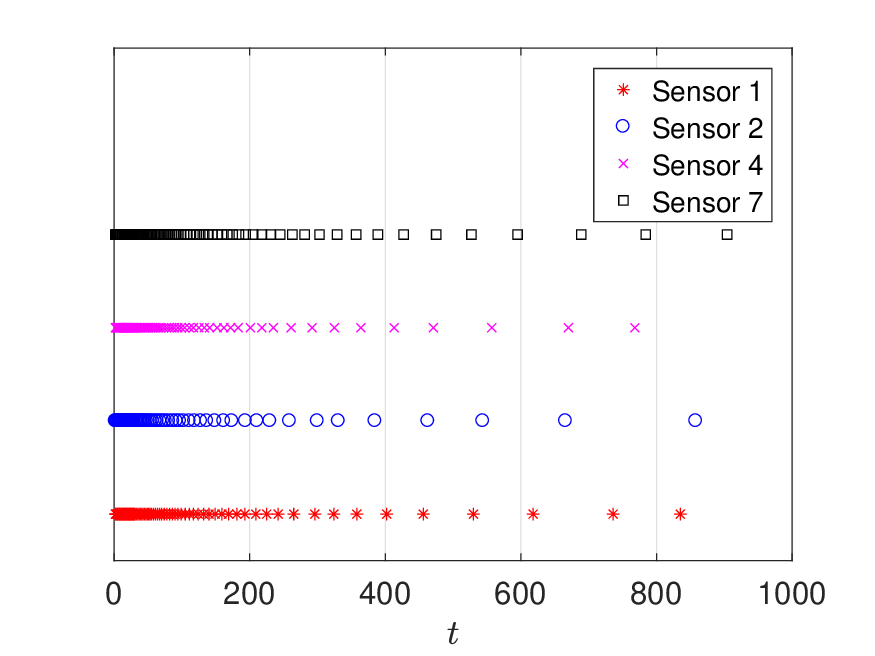}}
		\subfigure[\label{fig:rate}Dynamics of  communication rate $\lambda_c(t)$.]{\includegraphics[width=0.49\linewidth]{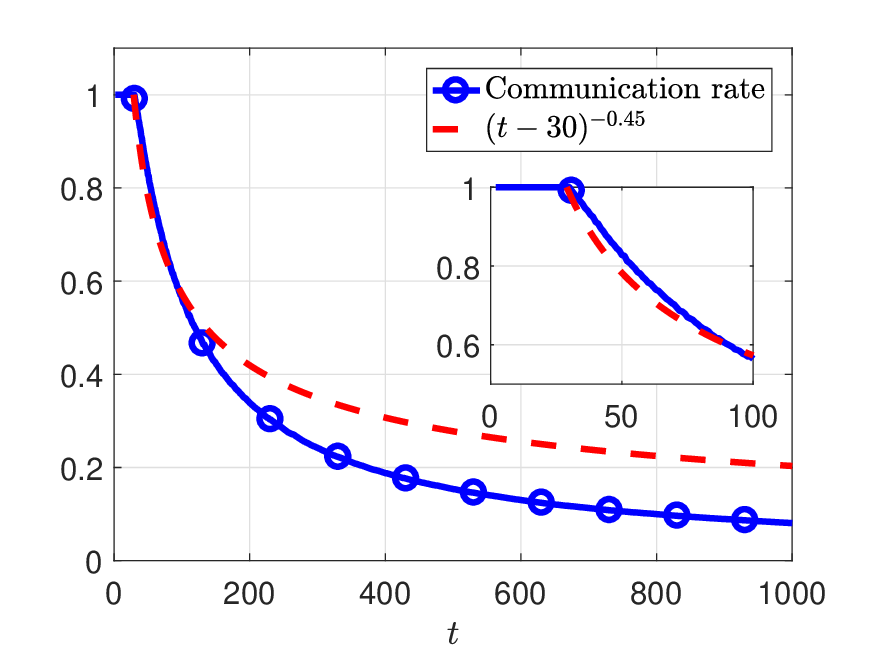}}
		\subfigure[\label{fig:error_ms}Mean-square convergence.]{\includegraphics[width=0.49\linewidth]{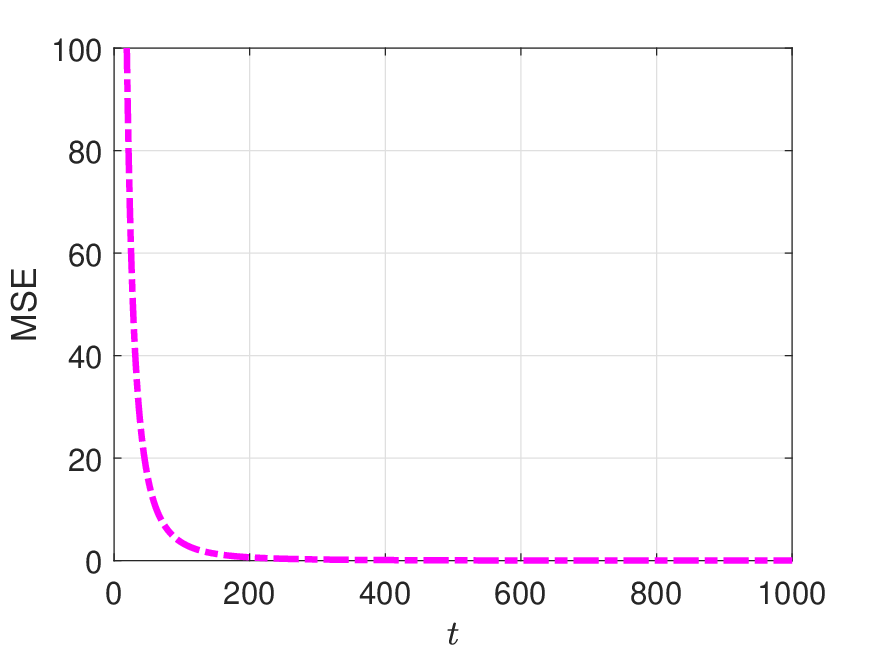}}
		\caption{\label{fig_instants} Simulation results of Algorithm \ref{alg:A}.}
	\end{figure*}

	\begin{figure*}
		\centering
		\subfigure[\label{fig:com_rate_threshold}Dynamics of  communication rate $\lambda_c(t)$.]{\includegraphics[width=0.49\linewidth]{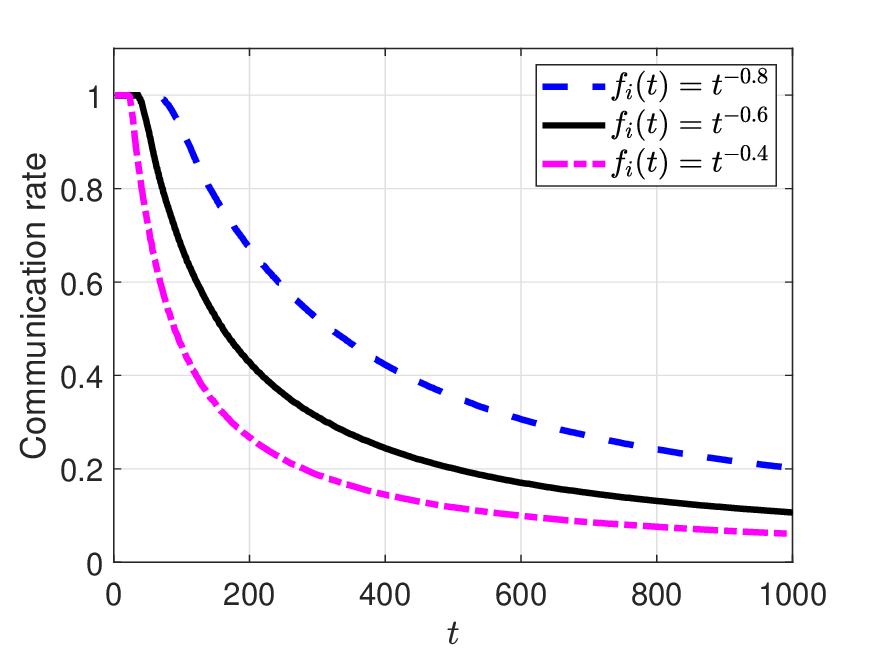}}
		\subfigure[\label{fig:error_mse_threshold}Mean-square convergence.]{\includegraphics[width=0.49\linewidth]{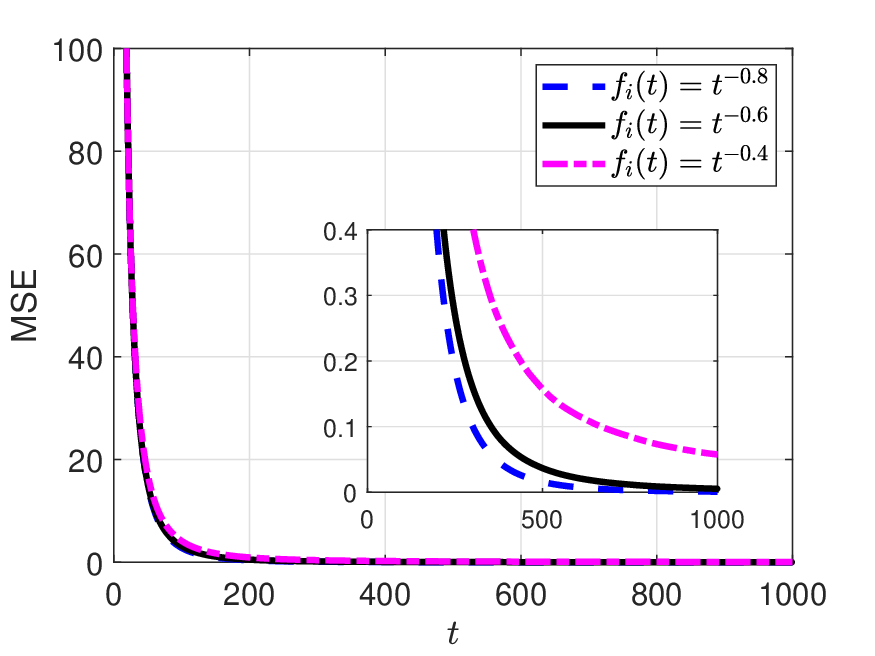}}
		\caption{\label{fig_thresholds} Communication rate and MSE of Algorithm \ref{alg:A} under three triggering thresholds.}
	\end{figure*}
	

	Under the above setting, we conduct a Monte Carlo experiment 
	with $M_0=100$ runs for Algorithm \ref{alg:A} with   $\alpha_i(t)=t^{-0.7}$, and $f_{i}(t)=t^{-0.5}$  for $i=1,2,\dots,7$. To evaluate the mean-square error,  we define
	\begin{align}\label{MSE_notation}
	\text{MSE}(t)&=\frac{1}{NM_0}\sum_{j=1}^{M_0}\sum_{i=1}^{N}\norm{x_i^j(t)-\theta}^2,
	\end{align}
	where $x_i^j(t)$ is the estimate of $\theta$ by sensor $i$ at time $t$ in the $j$-th run. The simulation results are provided in  FIG.~\ref{fig_instants}.  From FIG.~\ref{fig:times2}, the average estimate of all sensors is asymptotically convergent to the true parameter vector. The event-triggered communication triggering instants of sensors 1, 2, 4, and 7 are provided in FIG.~\ref{fig:times}, where we can see less and less communications  occur as time goes on. The communication rate is  0.08 in the interval $t=[0,1000]\cap \mathbb{N}$.  The dynamics of the communication rate in the given interval is provided in FIG.~\ref{fig:rate}, where  the communication rate remains to be 1 from $t=1$ to $t=30$ meaning that sensors persistently communicate with each other.  This is because at the initial time, much informative data can be used to update the sensor estimates.  As time goes on,  the communication rate is tending to zero, since sensors only transmits informative data  which is becoming less. Moreover, in order to illustrate the convergence rate of the communication rate, we provide the dynamics of $(t-30)^{-0.45}$ for $t>30$. From FIG.~\ref{fig:rate}, the communication rate  asymptotically decays  to zero faster than $(t-30)^{-0.45}$, corresponding to the results in Theorem \ref{thm_rate}.
	The mean-square convergence of the algorithm is illustrated in FIG.~\ref{fig:error_ms}, corresponding to  Theorem~\ref{thm_mean_square_convergence}.
	{Moreover, by choosing three triggering thresholds, i.e., $f_i(t)=t^{-0.8},t^{-0.6},t^{-0.4}$, we provide Fig.~\ref{fig_thresholds} for  illustrating the influence of triggering threshold to communication rate and MSE. From this figure, for the threshold with a faster decreasing speed, the corresponding communication rate decays more slowly, while the MSE decays more quickly. This indicates that the threshold   leads to a tradeoff between communication rate and MSE, corresponding to Remark~\ref{rem_tradeoff}. }

	\subsection{Example 2}
	In this example, we compare the proposed algorithm  with three existing algorithms over a   sensor network whose size is larger than that in example 1.
	Consider an undirected connected sensor network with 200 nodes for estimating a  target position  $\theta=[1,2,5]$. The sensor network topology is generated as a random geometric graph
	and provided in FIG. \ref{fig:network2}, where two types of sensors are deployed with the same number (i.e., 100) and denoted by black circle and red  diamond.  Assume  for any $i=1,2,\dots,200$, the weight  $a_{i,j}=1$,  for $ j\in \mathcal{N}_i$.  The measurement matrices of black-circle and red-diamond sensors are assumed to be $H_A=[0,0,1]$ and  $H_B=[1,0,0;0,1,0]$, respectively. Suppose the sensor noise follows a   standard Gaussian process and is independent in both space and time.
	
	\begin{figure*}[t]
		\centering
		\subfigure[\label{fig:network2}Sensor network.]{\includegraphics[width=0.49\linewidth]{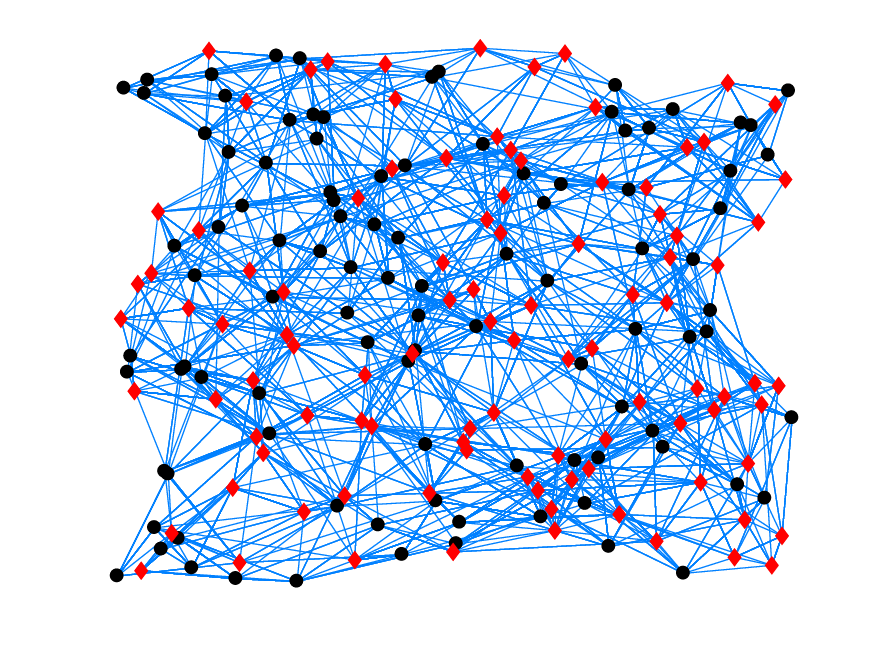}}
		\subfigure[\label{fig:compare} MSE performance of four algorithms.]{\includegraphics[width=0.49\linewidth]{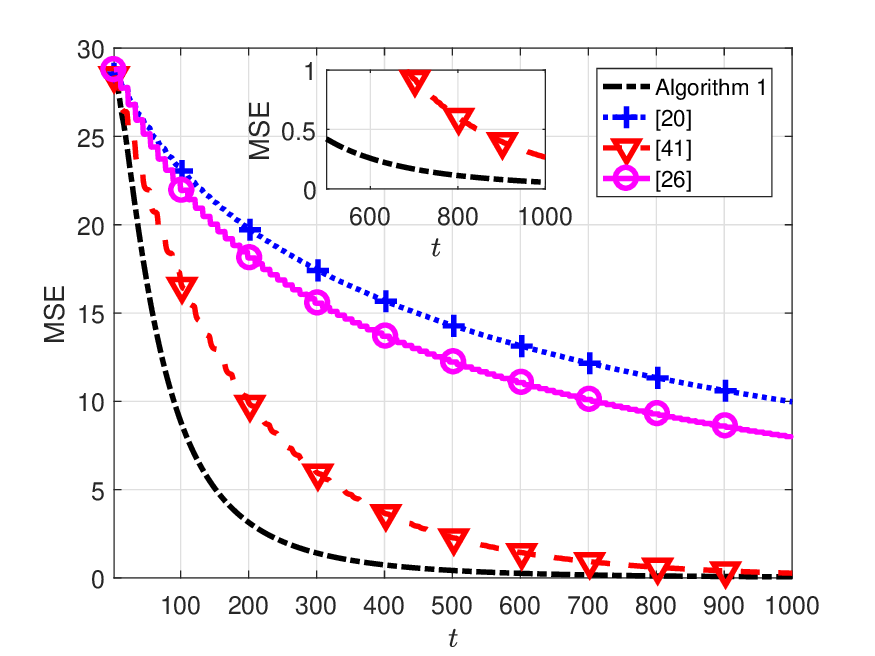}}
		\caption{\label{fig_example2} Comparison of four distributed algorithms over a sensor network with 200 nodes.}
	\end{figure*}
	
	We compare Algorithm \ref{alg:A} with three typical  distributed estimation algorithms in the literature, including the  generalized linear unconstrained algorithm from \cite{kar2011convergence},  the diffusion least-mean squares algorithm from \cite{lopes2008diffusion}, and  the distributed parameter algorithm from \cite{zhang2012distributed}. The parameters in our algorithm are $\alpha_i(t)=(t+100)^{-0.7}$, and $f_{i}(t)=t^{-0.5}$  for $i=1,2,\dots,200$.
	The parameters in the algorithm of \cite{kar2011convergence} are $\alpha(t)=10/(t+1)^{0.7}$, $\beta(t)=0.1/(t+1)^{0.7}$, and $K=(\sum_{i=1}^{200}H_i^TH_i)^{-1}$, where $H_i$ is the measurement matrix of sensor $i$. The parameters in the algorithm of \cite{zhang2012distributed} are $b(t)=(t+100)^{-0.7}$ and $a_{i,j}=1$, for $ i,j=1,2,\dots,200$.  The parameters in the algorithm of \cite{lopes2008diffusion} are $\mu(t)=(t+100)^{-0.7}$ and $c_{i,j}=1/|\mathcal{N}_i+1|$ for $j\in\mathcal{N}_i\bigcup \{i\}$, $i=1,2,\dots,200$. The initial parameter estimate is  zero for each algorithm.
	We compare our event-triggered algorithm with the above three time-triggered algorithms   under the same communication rate $\lambda_{c}=0.09$. It  means that for the three time-triggered algorithms, each sensor   receives the  messages from   neighbors for every $11\approx1/\lambda_{c}$ steps, before that they use the latest  messages from   neighbors to run the algorithms.
	Under this setting, we conduct a Monte Carlo experiment for running the four algorithms simultaneously with $100$ runs. With the MSE notation in \eqref{MSE_notation}, the performance of these algorithms are illustrated in FIG. \ref{fig:compare}. It shows that  the outputs of all algorithms are convergent to the true parameter vector, and   our event-triggered distributed algorithm outweighs other three algorithms in convergence speed under the same communication rate constraint.

	\section{Conclusion}\label{sec:conclusion}
	In this paper, a distributed parameter estimation problem over a  sensor network with event-triggered communications was studied.
	First,     a fully distributed estimation algorithm was proposed based on an event-triggered communication scheme  which determines when a sensor should  share the parameter  estimates with neighboring sensors. 
	Then, under  mild conditions, some main estimation properties of the algorithm including   mean-square  and almost-sure convergence  were analyzed, respectively. The convergence rates were also estimated.  Under some extra conditions, it was proved that the communication rate of the whole network using the proposed algorithm decays to zero almost surely as time goes to infinity, which indicates that a tremendous amount of redundant communications are avoided. 
	It was also shown that   adjusting the decay speed of the triggering threshold can lead to a tradeoff between the convergence rate of the estimation error and the decay speed of the communication rate.
	Future work can be done by considering more general  models of systems and networks, such as nonlinear measurement models,  {time-varying, unbalanced, or stochastic  graphs}.

\bibliography{All_references}
\bibliographystyle{ieeetr}

%
%

%
%
%

\end{document}